\documentclass{sig-alternate}

\PassOptionsToPackage{table}{xcolor}
\usepackage{graphicx}
\usepackage{stmaryrd}
\usepackage{tikz}
\usetikzlibrary{shapes,arrows,shapes.geometric,fit,matrix,backgrounds}
\PassOptionsToPackage{table}{xcolor}
\usepackage{graphicx}
\usepackage{caption}
\usepackage{subcaption}
\usepackage{framed}
\usepackage{fancybox}
\usepackage{colortbl}
\usepackage{color}
\usepackage{hyperref}
\usepackage{pgfpages}
\usepackage{amsmath,mathtools,amssymb}
\usepackage{balance}  
\usepackage{epigraph}
\usepackage{lineno}
\usepackage{algorithm}
\usepackage{algorithmicx}
\usepackage[noend]{algpseudocode}
\newtheorem{myex}{Example}

\newtheorem{mydef}{Def.}
\newtheorem{myprob}{Problem}
\newtheorem{myconj}{Conjecture}
\newtheorem{mythm}{Theorem}
\newtheorem{myprop}{Proposition}
\newtheorem{mycor}{Corollary}
\newtheorem{mylemma}{Lemma}
\newtheorem{myremark}{Remark}

\begin{document}
%

\title{Design-theoretic encoding of deterministic hypotheses as constraints and correlations into U-relational databases}

\numberofauthors{2} 
\author{
%
\alignauthor
Bernardo Gon\c{c}alves\\
       \affaddr{LNCC -- National Laboratory}\\
       \affaddr{for Scientific Computing}\\
       \affaddr{Petr\'opolis, Brazil}\\
       \email{bgonc@lncc.br}
\alignauthor Fabio Porto\\
       \affaddr{LNCC -- National Laboratory}\\
       \affaddr{for Scientific Computing}\\
       \affaddr{Petr\'opolis, Brazil}\\
       \email{fporto@lncc.br}
}

\maketitle
\fontsize{10pt}{10pt}
\selectfont

\begin{abstract}
\fontsize{10pt}{10pt}
\selectfont
In view of the paradigm shift that makes science ever more data-driven, in this paper we consider deterministic scientific hypotheses as uncertain data. In the form of mathematical equations, hypotheses symmetrically relate aspects of the studied phenomena. For computing predictions, however, deterministic hypotheses are used asymmetrically as functions. We refer to Simon's notion of structural equations in order to extract the (so-called) causal ordering embedded in a hypothesis. Then we encode it into a set of functional dependencies (fd's) that is basic input to a design-theoretic method for the synthesis of U-relational databases (DB's). 

The causal ordering captured from a formally-specified system of mathematical equations into fd's determines not only the constraints (structure), but also the correlations (uncertainty chaining) hidden in the hypothesis predictive data. We show how to process it effectively through original algorithms for encoding and reasoning on the given hypotheses as constraints and correlations into U-relational DB's. The method is applicable to both quantitative and qualitative hypotheses and has underwent initial tests in a realistic use case from computational science.
\end{abstract}


\vspace{-3pt}

\section{Introduction}
\noindent
As part of the paradigm shift that makes science ever more data-driven, deterministic scientific hypotheses can be seen as: principles or ideas, which are mathematically expressed and then implemented in a program that is run to give their \emph{decisive} form of data. For a description of the research vision of hypothesis management and its significance, we refer the reader to \cite{goncalves2014}.\footnote{Also, to anticipate \S\ref{subsec:related-work}, it can be understood in comparison as a specific, shorter-term research path along the lines of Haas et al.'s \emph{models-and-data} program \cite{haas2011}.} In this paper we engage in a theoretical exploration on deterministic hypotheses as a kind of uncertain data.

\textbf{Target applications.} 
Our framework is geared at hypothesis management applications. Examples of structured deterministic hypotheses include tentative mathematical models in physics, engineering and economical sciences, or conjectured boolean networks in biology and social sciences. These are important reasoning devices, as they are solved to generate predictive data for decision making in both science and business. But the complexity and scale of modern scientific problems require proper data management tools for the predicted data to be analyzed more effectively.

\textbf{Probabilistic DBs.} 
Probabilistic databases (p-DBs) qualify as such tool, as they have evolved into mature technology in the last decade \cite{suciu2011}. One of the state-of-the-art probabilsitic data models is the U-relational representation system with its probabilistic world-set algebra (p-WSA) implemented in \textsf{MayBMS} $\!$\cite{koch2009}. $\!$That is an elegant extension of the relational model we refer to in this paper for the management of uncertain and probabilistic data. Our goal is to develop means to extract a hypothesis specification and encode it into a U-relational DB seamlessly, ensuring consistency and quality w.r.t.$\!$ the given hypothesis structure. 
In short, we shall flatten deterministic models into U-relations. 

\textbf{Structural equations.} 
Given a system of equations with a set of variables appearing in them, in a seminal article Simon introduced an asymmetrical, functional relation among variables that establishes a (so-called) \emph{causal ordering} \cite{simon1953}. Along these lines, we shall extract the causal ordering of a deterministic hypothesis and encode it into a set of fd's that is basic input to our synthesis of U-relational DBs. As we shall see, the causal ordering we capture in fd's determines not only the constraints (structure), but also the correlations (uncertainty chaining) hidden in predictive data. 
In comparison with research on causality in DBs \cite{meliou2010} (cf. \S\ref{subsec:related-work}), our framework comprises a technique for encoding and processing causality at schema level.

\textbf{Background theory.} 
We rely on the following body of background theoretical work: (i) U-relations and p-WSA \cite{koch2009}, as our design-theoretic method is shaped for U-relational DBs; (ii) classical theory of fd's and normalization 
\cite{abiteboul1995,ullman1988}, and Bernstein's design-by-synthesis approach \cite{bernstein1976}; and finally, (iii) Simon's notion of structural equation models (SEMs) \cite{pearl2000,simon1953}.

$\!\!$\textbf{List of contributions.} $\!$Overall, this paper presents specific techni\-cal developments over the $\!\Upsilon\!$-DB vision $\!$\cite{goncalves2014}. $\!$In short, it shows how to encode deterministic hypotheses as uncertain data, viz., as constraints and correlations into U-relational DBs. Our detailed technical contributions are:
\begin{itemize}\vspace{-3pt}

\item We study the relationship between SEMs and fd's and present an \emph{encoding scheme} that extracts the causal ordering of a given (formally specified) deterministic hypothesis and encodes it into a set $\Sigma$ of fd's; we then uncover some of its main properties. In short, given a hypothesis, we show how to transform it algorithmically into a set $\Sigma$ of fd's in order to design a relational DB over it.\vspace{-1pt}

\item For a ``good'' design, we show that the hypothesis causal ordering mapped into fd set $\Sigma$ needs to be further processed in terms of \emph{acyclic reasoning on its reflexive pseudo-transitive closure}. $\!$We present an original, efficient algo\-rithm for that, which returns an fd set $\Sigma^\prime$ we motivate and define to be the \emph{folding} $\Sigma^\looparrowright\!$ of $\Sigma$. Then we apply a variant of Bernstein's synthesis algorithm (say, `4C') to render, given $\Sigma^\looparrowright\!$, a relational schema $\boldsymbol H_k\!= \bigcup_{i=1}^n H_k^i$
 shown to bear desirable properties for hypothesis management --- yet up to the capabilities of a traditional relational DB at this stage of the design pipeline. 
In short, this is a design-theoretic technique that uses the extracted fd's as constraints.\vspace{-1pt}

\item Once schema $\boldsymbol H_k$ is synthesized, datasets computed from the hypothesis under alternative trials (input settings) can be loaded into it. Finally, then, through a different manipulation on the primitive fd set $\Sigma$, we extract the \emph{uncertainty chaining} (correlations) from it into a $\Sigma^{\prime\prime}$ which is, together with $\boldsymbol H_k$, input to an original synthesis procedure (say, `4U') to render U-relational $\boldsymbol Y_k\!=\! \bigcup_{j=1}^m Y_k^j$. In short, this is a principled technique to \emph{introduce uncertainty} in  p-WSA given a set of fd's.\vspace{-1pt}
\end{itemize}

After introducing notation and basic concepts in $\S\ref{sec:preliminaries}$, we present through \S\ref{sec:encoding}--\S\ref{sec:synthesis4u} the contributions (resp.) listed above. We discuss related work and the applicability of our framework in \S\ref{sec:discussion}, and also point to a use case scenario from which we have extracted and encoded real-world hypotheses and conducted some initial experiments. Finally, \S\ref{sec:conclusions} concludes the paper.

\section{Preliminaries}\label{sec:preliminaries}
\noindent
As notational conventions, we write $X, Y, Z$ to denote sets of relational attributes and $A, B,$ $C$ to denote single attributes. Also, we write $XY$ as shorthand for $X \cup Y$, and $R[XZ]$ to denote relation $R$ has scheme $U\!=\!XZ$ with designated key constraint $X \!\to Z$.

\subsection{U-Relations and Probabilistic WSA}\label{subsec:u-relations}
\noindent
A U-relational DB or U-DB is a finite set of structures,
\begin{center}
$\boldsymbol W \!=\! \{\langle R_1^1, \hdots, R_m^1, p^{[1]}\rangle, \hdots, \langle R_1^n, \hdots, R_m^n, p^{[n]}\rangle\}$,
\end{center}
 of relations $R_1^i, \hdots, R_m^i$ and numbers $0 < p^{[i]} \leq 1$ such that
$\sum_{1\leq i \leq n} p^{[i]} = 1$. An element $R_1^i, \hdots, R_m^i, p^{[i]} \in \boldsymbol W$ is a \emph{possible world}, with $p^{[i]}$ being its probability \cite{koch2009}.

Probabilistic world-set algebra (p-$\!$WSA) consists of the operations of relational algebra, an operation for computing tuple confidence \textsf{conf}, and the \textsf{repair-key} operation for introducing uncertainty --- by giving rise to alternative worlds as maximal-subset repairs of an argument key (cf. Def. \ref{def:repair-key}) \cite{koch2009}.

\begin{mydef}
Let $R_\ell[U]$ be a relation, and $XA \subseteq U$. For each possible world $\langle R_1, \hdots, R_m, p\rangle \in \boldsymbol W$, let $A \in U$ contain only numerical values greater than zero and let $R_\ell$ satisfy the fd $(U \setminus A) \to U$. Then, \textsf{repair-key} is:

\vspace{4pt}
\noindent
$\!\!\llbracket \textsf{repair-key}_{X@A}(\!R_\ell) \rrbracket (\boldsymbol W\!) \!:=\! \left\{ \!\langle R_1,  .., R_\ell, R_m, \hat{R_\ell}[U\!\setminus\!A], \hat{p}\rangle \!\right\}$,
\vspace{2pt}

\noindent
where $\langle R_1, ..., R_\ell, R_m, p \rangle \in \boldsymbol W$, $\hat{R}_\ell$ is a maximal repair of fd $X \to U$ in $R_\ell$, and $\hat{p} = p\, \cdot\!\! \displaystyle\prod_{t \in \hat{R}_\ell} \frac{t.B}{\sum_{s \in R_\ell : s.X=t.X} s.B}$.
\label{def:repair-key}
\end{mydef}

\vspace{-8pt}
\noindent
U-relations (cf. Fig. \ref{fig:maybms}) have in their schema a set of pairs $(V_i, D_i)$ of \emph{condition columns} (cf.$\!$ \cite{koch2009}) to map each discrete random variable $\textsf{x}_i$ to one of its possible values \mbox{(e.g., $\textsf{x}_1 \!\mapsto\! 1$)}. The world table $W\!$ stores their mar\-ginal probabilities (cf. $\!$ the notion of \emph{pc-tables} \cite[Ch. $\!$2]{suciu2011}). 
For an illustration of the data transformation from certain to uncertain relations, consider query (\ref{eq:explanation}) in p-WSA's extension of relational algebra, whose result set is materialized into U-relation \textsf{Y}$_0$ as shown in (Fig. \ref{fig:maybms}).

\vspace{-10pt}
\begin{eqnarray}
\! Y_0 \,:=\, \pi_{\phi,\upsilon}( \textsf{repair-key}_{\phi @\textsf{Conf}} (H_0)\,).
\label{eq:explanation}
\end{eqnarray}
\noindent
Also, let $R[\,\overline{V_i\,D_i} \;|\, sch(R)\,],\, S[\,\overline{V_j\,D_j} \;|\, sch(S)\,]$ be two U-relations, where $R.\,\overline{V_i\,D_i}$ is the union of all pairs of condition columns $V_i\,D_i$ in $R$, then operations $\llbracket\, \sigma_\psi(R)\,\rrbracket$, $\llbracket\, \pi_Z(R) \,\rrbracket$ and $\llbracket R \times S \rrbracket$ issued in relational algebra are rewritten in positive relational algebra on U-relations:\vspace{-4pt}\\

$\llbracket \sigma_\psi(R) \rrbracket := \sigma_\psi(R[\overline{V_i\,D_i}\,\,|\, sch(R)])$;\vspace{3pt}

$\llbracket\, \pi_Z(R) \,\rrbracket := \pi_{\,\overline{V_i\,D_i}\,Z}(R)$;\vspace{3pt}

$\llbracket R \times S \rrbracket := \pi_{(R.\overline{V_i\,D_i} \;\cup\; S.\overline{V_i\,D_i}) \to \overline{V\,D} \;\cup\; sch(R) \,\cup\, sch(S)} ( R$\\ $\bowtie_{R.\overline{V_i\,D_i} \;\textsf{is consistent with}\; S.\overline{V_j\,D_j}} S )$.

\vspace{6pt}

If $R$ and $S$ have $k$ and $\ell$ pairs of condition columns each, then $\llbracket R \times S \rrbracket$ returns a U-relation with $k + \ell$ such pairs. If $k=0$ or $\ell=0$ (or both), then $R$ or $S$ (or both) are classical relations, but the rewrite rules above apply accordingly. All that rewriting is parsimonious translation (sic. \cite{koch2009}): the number of algebraic operations does not increase and each of the operations selection, projection and product/join remains of the same kind. Query plans are hardly more complicated than the input queries. In fact, off-the-shelf relational database query optimizers do well in practice.

For a comprehensive overview of U-relations and p-WSA we refer the reader to \cite{koch2009}. In this paper we look at U-relations from the point of view of p-DB design, for which no methodology has yet been proposed. We are concerned in particular with hypothesis management applications \cite{goncalves2014}.

\begin{figure}[t]
\centering
\footnotesize
\begingroup\setlength{\fboxsep}{2pt}
\colorbox{blue!7}{%
   \begin{tabular}{c|c|c|c}
  \textsf{H}$_0$ & $\phi$ & $\upsilon$ & \textsf{Conf}\\
      \hline    
   & $1$ & $1$ & 2\\
   & $1$ & $2$ & 2\\
   & $1$ & $3$ & 1\\
   \end{tabular}
}\endgroup\vspace{2pt}\\
\begingroup\setlength{\fboxsep}{2pt}
\colorbox{yellow!15}{%
   \begin{tabular}{c|>{\columncolor[gray]{0.92}}c||c|c}
  \textsf{Y}$_0$ & $V \mapsto D$ & $\phi$ & $\upsilon$\\
      \hline    
   & $\textsf{x}_0 \mapsto 1$ & $1$ & $1$\\
   & $\textsf{x}_0 \mapsto 2$ & $1$ & $2$\\
   & $\textsf{x}_0 \mapsto 3$ & $1$ & $3$\\
   \end{tabular}
}\endgroup
\begingroup\setlength{\fboxsep}{2pt}
\colorbox{yellow!15}{%
   \begin{tabular}{c|>{\columncolor[gray]{0.92}}c||c}
  \textsf{W} & $V \mapsto D$ & \textsf{Pr}\\
      \hline    
   & $\textsf{x}_0 \mapsto 1$ & $.4$\\
   & $\textsf{x}_0 \mapsto 2$ & $.4$\\
   & $\textsf{x}_0 \mapsto 3$ & $.2$\\
   \end{tabular}
}\endgroup
\vspace{-3pt}
\caption{U-relation generated by the repair-key operation.}
\label{fig:maybms}
\vspace{-9pt}
\end{figure}

\subsection{Design by Synthesis and Normalization}\label{subsec:synthesis}

\noindent
The problem of design by synthesis has long been introduced by Bernstein in purely symbolic terms as follows \cite{bernstein1976}: 
given a set $U\!$ of attribute symbols and a set $\Sigma$ of mappings of sets of symbols into symbols (the fd's), find a collection $\boldsymbol R \!=\! \{R_1, R_2, \!\hdots, R_n\}$ (the relations) of subsets of $U$ and, for each $R_i$, a subset of $R_i$ (its designated key) satisfying properties: (P1) each $R_i \in \boldsymbol R$ is in 3NF; (P2) $\boldsymbol R$ completely characterizes $\Sigma$; and (P3) the cardinality $|\boldsymbol R|$ is minimal. 

More generally, the problem of schema design given dependencies considers the following criteria \cite[Ch. $\!$11]{abiteboul1995}:
\begin{itemize}
\item[P1$^\prime\!$.] $\!\!$($\simeq\!$ P1). $\!$Desirable properties by normal forms;\vspace{-5pt}
\item[P2$^\prime\!$.] $\!\!$($\simeq\!$ P2). $\!$Preservation of dependencies (``meta-data'');\vspace{-5pt}
\item[P3.] $\!$The cardinality $|\boldsymbol R|$ is minimal (minimize joins);\vspace{-5pt}
\item[P4.] $\!$Preservation of data (the lossless join property).

\end{itemize}

There is a trade-off between P1$^\prime$ and P2$^\prime$, since normal forms that ensure less redundant schemes may lose the property of dependency preservation \cite{abiteboul1995}. In fact, P2$^\prime$ is important to prevent the DB from the so-called update anomalies, as the fd's in $\Sigma$ are viewed as integrity constraints to their associated relations \cite[p. $\!$398]{ullman1988}. Hypothesis management applications \cite{goncalves2014}, however, are OLAP-like and have an ETL-pipeline characterized by batch-, incremental-only updates and large data volumes. 
$\!$Thus we \mbox{shall trade P2$^\prime$ for P1$^\prime$, to favor} succintness (as less redundancy as possible) over dependency preservation (recall BCNF, Def. \ref{def:nf}).  
Also, we shall favor P4 as less joins means faster access to data.

Recall from Ullman \cite{ullman1988} that an attribute $A \!\in U$ is said to be \emph{prime} in relation schema $R[U]$ if it is part of some key for $R[U]$. Def. \ref{def:nf} presents the Boyce-Codd normal form (BCNF) and the third normal form (3NF).

\vspace{-1pt}
\begin{mydef}
$\!$Let $R[U]$ be a relation scheme over set $U\!$ of attributes, and $\Sigma$ a set of fd's on $U$.  We say that:%
\begin{itemize}
\item[(a)] $R$ is in \textbf{BCNF} if, for all $\langle X, A\rangle \in \Sigma^+\!$ with $A \!\nsubseteq X$ and $XA \!\subseteq\! U$, we have $X \!\to U$ (i.e., $\!X$ is a superkey for $R$);\vspace{-2pt}

\item[(b)] $R$ is in \textbf{3NF} if, for all $\langle X, A\rangle \in \Sigma^+$ with $A \nsubseteq X$ and $XA \!\subseteq\! U$, we have $X \!\to U$ or $A$ is prime;\vspace{-2pt}

\item[(c)] A schema $\boldsymbol R$ $\!$is in BCNF (3NF) w.r.t. $\!\Sigma$ if $\!$all of its schemes $R_1, ..., R_n \in \boldsymbol R$ are in BCNF (3NF).
\end{itemize}
\label{def:nf}
\end{mydef}
\noindent
We also recall from \cite{ullman1988} the lossless join property (Def. $\!\!$\ref{def:lossless-join}) and the notion of dependency preservation (Def. \ref{def:fd-preserving}). %

\begin{mydef}
Let $R[U]\!$ be a relational schema synthesized into collection $\boldsymbol R = \bigcup_{i=1}^n R_i$ and let $\Sigma$ be an fd set on attributes $U$. We say that $\boldsymbol R$ has a \textbf{lossless join} w.r.t. $\!\Sigma$ if for every instance $r$ of $R[U]$ satisfying $\Sigma$, we have $r =\; \bowtie_{i=1}^n \pi_{R_i} (r)$.
\label{def:lossless-join}
\end{mydef}
\vspace{-5pt}
\begin{mydef}
Let $\Sigma$ be a set of fd's and $\boldsymbol R \!=\! \{R_1, R_2, ..., R_n\}$ be a relational schema. We say that $\boldsymbol R$ \textbf{preserves} $\Sigma$ if the union of all fd's in $\boldsymbol R\!=\! \{R_1, R_2, \hdots, R_n\}$ implies $\Sigma$.
\label{def:fd-preserving}
\end{mydef}

 \vspace{-9pt}
Design theory and normalization relies on Armstrong's inference rules (or axioms) of (R0) reflexivity, (R1) augmentation and (R2) transitivity, which forms a sound and complete inference system for reasoning over fd's \cite{ullman1988}. From R0-R2 one can derive additional rules, viz., (R3) decomposition, (R4) union and (R5) pseudo-transi\-tivity. 
 \vspace{3pt}\\ 
\indent
\textbf{R0}. If $Y \!\subseteq X$, then $X \!\to Y$;\vspace{1pt}\\
\indent
\textbf{R1}. If $X \!\to Y$, then $XZ \!\to YZ$;\vspace{1pt}\\
\indent
\textbf{R2}. $\!$If $X \!\to Y\!$ and $Y \!\to W\!$, then $X \!\to W$;\vspace{1pt}\\
\indent
\textbf{R3}. If $X \!\to YZ$, then $X \!\to Y$ and $X \!\to Z$;\vspace{1pt}\\
\indent
\textbf{R4}. If $X \!\to Y$ and $X \!\to Z$, then $X \!\to YZ$;\vspace{1pt}\\
\indent
\textbf{R5}. $\!$If $X \!\to Y\!$ and $YZ \!\to W\!$, then $XZ \!\to W$.
\vspace{5pt}

\noindent
Given an fd set $\Sigma$, one can obtain $\Sigma^+$, the closure of $\Sigma$, by a finite application of rules R0-R5. We are concerned with reasoning over an fd set in order to process its `embedded' causal ordering. The latter, as we shall see in \S\ref{sec:synthesis4c}, can be performed in terms of reflexive \mbox{(pseudo-)}transitive reasoning. Note that R2 is a particular case of R5 when $Z\!=\!\varnothing$, then we shall refer to $\{$R0, $\!$R5$\}$ reasoning and understand R2 included. Def. \ref{def:xclosure} opens up a way to compute $\Sigma^+\!$ efficiently.

\begin{mydef}
Let $\Sigma$ be an fd set on attributes $U$, with $X \subseteq U$. Then $X^+$, the \textbf{attribute closure} of $X$ w.r.t. $\Sigma$, is the set of attributes $A$ such that $\langle X, A\rangle \in \Sigma^+$.
\label{def:xclosure}
\end{mydef}

\noindent
Bernstein has long given algorithm \textsf{XClosure} (cf. $\!$Alg. $\!$\ref{alg:xclosure}$\,$ in \S\ref{algs}) to compute $X^+$ in time that is polynomial in $|\Sigma| \!\cdot |U|$ \cite{bernstein1976}. 
Finally, we shall also make use of the concept of `canonical' fd sets (also called `minimal' \cite[p.$\!$ 390]{ullman1988}), see Def. $\!$\ref{def:minimal}.
\begin{mydef}
$\!\!$Let $\Sigma$ be an fd set. $\!\!$We$\!$ say that $\Sigma$ is \textbf{canonical} if:
\begin{itemize}
\vspace{-2pt}
\item[(a)] each fd in $\Sigma$ has the form $X \!\to A$, where $|A|=1$;\vspace{-2pt}
\item[(b)] For no $\langle X, A\rangle \in \Sigma$ we have $(\Sigma - \{\langle X, A\rangle\})^+ = \Sigma^+$;\vspace{-2pt}
\item[(c)] for each fd $X \!\to A$ in $\Sigma$, there is no $Y \subset X$ such that $(\Sigma \setminus \{X \!\to A\} \cup \{Y \!\to A \})^+ = \Sigma^+$.
\end{itemize}
\label{def:minimal}
\end{mydef}
For an fd set satisfying such properties (Def. $\!$\ref{def:minimal}) individually, we say that it is (a) \emph{singleton-rhs}, (b) \emph{non-redundant} and (c) \emph{left-reduced}. It is said to have an attribute $A$ in $X$ that is `extraneous' w.r.t. $\!\Sigma$ if it is \emph{not} left-reduced (Def. \ref{def:minimal}-c) \cite[p. $\!$74]{maier1983}. Finally, an fd $X \!\to Y$ in $\Sigma$ is said \emph{trivial} if $Y \subseteq X$.

\subsection{SEMs and Causal Ordering}\label{subsec:coa}
\noindent
Given a system of mathematical equations involving a set of variables, to build a \emph{structural equation model} (SEM) is, essentially, to establish a one-to-one mapping between equations and variables \cite{simon1953}. That enables further detecting the hidden asymmetry between variables, i.e., their causal ordering. 

\begin{mydef}
A \textbf{structure} is a pair $\mathcal S(\mathcal E, \mathcal V)$, where $\mathcal E$ is a set of equations over set $\mathcal V\!$ of variables, $|\mathcal E| \leq |\mathcal V|$, such that:
\begin{itemize}
\item[(a)] In any subset of $k$ equations of the structure, at least $k$ different variables appear;

\item[(b)] In any subset of $k$ equations in which $r$ variables appear, $k \leq r$, if the values of any $(r - k)$ variables are chosen arbitrarily, then the values of the remaining $k$ variables can be determined uniquely --- finding these unique values is a matter of solving the equations.
\end{itemize}
\end{mydef}

\begin{mydef}
Let $\mathcal S(\mathcal E, \mathcal V)$ be a structure. We say that $\mathcal S$ is self-contained or \textbf{complete} if $|\mathcal E|=|\mathcal V|$.
\label{def:complete}
\end{mydef}

Complete structures can be solved for unique sets of values of their variables. In this work, however, we are not concerned with solving sets of mathematical equations at all, but with extracting their causal ordering in view of U-relational DB design. Simon's concept of causal ordering has its roots in econometrics studies (cf. \cite{simon1953}) and to some extent has been taken further in AI with a flavor of Graphical Models (GMs) \cite{druzdzel1993,pearl2000,druzdzel2008}. In this paper we translate the problem of causal ordering into the language of data dependencies, viz., into fd's.

\begin{mydef}
Let $\mathcal S$ be a structure. We say that $\mathcal S$ is \textbf{minimal} if it is complete and there is no complete structure $\mathcal S^\prime \!\subset \mathcal S$. 
\end{mydef}

\begin{mydef}
The \textbf{structure matrix} $A_S$ of a structure $\mathcal S(\mathcal E, \mathcal V)$, with $f_1,\, f_2, \hdots, f_n \in \mathcal E$ and $x_1,\, x_2, \hdots, x_m \in \mathcal V$, is a $n \times m$ matrix of 1's and 0's in which entry $a_{ij}$ is non-zero if variable $x_j$ appears in equation $f_i$, and zero otherwise.
\end{mydef}

Elementary row operations (e.g., row multiplication by a constant) on the structure matrix may hinder the structure's causal ordering and then are not valid in general \cite{simon1953}. This also emphasizes that the problem of causal ordering is not about solving the system of mathematical equations of a structure, but identifying its hidden asymmetries.

\begin{mydef}
$\!$Let $\mathcal S(\mathcal E, \mathcal V)\!$ be a complete structure. $\!\!$Then a \textbf{total causal mapping} over $\mathcal S$ is a
bijection $\varphi_t\!: \mathcal E \to \mathcal V$.
\label{def:total-causal-mapping}
\end{mydef}

\vspace{-3pt}
\noindent
Simon has informally described an algorithm (cf. \cite{simon1953}) that, given a complete structure $\mathcal S(\mathcal E, \mathcal V)$, can be used to compute a \emph{partial} causal mapping $\varphi_p\!: \mathcal E \to \mathcal V$ from the set of equations to the set of variables. As shown by Dash and Druzdzel \cite{druzdzel2008}, the causal mapping returned by Simon's (so-called) Causal Ordering Algorithm (\textsf{COA}) is not \emph{total} when $\mathcal S$ has variables that are \emph{strongly coupled}, i.e., can only be determined simultaneously. They also have shown that any total mapping $\varphi_t$ over $\mathcal S$ must be consistent with \textsf{COA}'s partial mapping $\varphi_p$ \cite{druzdzel2008}. The latter is made partial by design (merge strongly coupled variables) to force its induced causal graph $G_{\varphi_p}$ to be acyclic. 

Dash and Druzdzel's work (cf. $\!$\cite{druzdzel2008}) is focused on the correcteness of \textsf{COA}, from a GM point of view. Instead, we shall elaborate on \textsf{COA} in purely symbolic terms, towards encoding structures into fd sets and reasoning over them using Armstrong's rewrite rules R0, R5.
For extracting a structure's causal ordering into an fd set, we are only concerned with total causal mappings and then shall have to deal with the issue of cyclic fd's in the causal ordering. We shall map (injectively) variables to relational attributes and (bijectively) equations to fd's.

\subsection{Problem Statement}\label{subsec:problem}
\vspace{-3pt}
\noindent
Now we can formulate more precisely the problems in our design pipeline (Fig. $\!$\ref{fig:pipeline}). In the `local' view for a hypothesis $k$, it synthesizes U-relations $\bigcup_{j=1}^{m}Y_k^j$ given its complete structure $\mathcal S_k$ and its alternative trial datasets $\bigcup_{\ell=1}^{p}\mathcal{D}_k^\ell$. In fact, the \textsf{U-intro} procedure is operated by the pipeline in the `global' view of all available hypotheses $k=1..z$. Their conditioning in the presence of evidence (cf. \cite{goncalves2014}) is not covered in this paper.

\begin{myprob} (\textbf{Hypothesis encoding}).
Given the (complete) structure $\mathcal S_k$ of deterministic hypothesis, extract a total causal mapping $\varphi_t$ over $\mathcal S_k$ and encode $\varphi_t$ into an fd set $\Sigma_k$.
\label{prob:h-encode}
\end{myprob}

Following the encoding of a hypothesis structure $S$ into a set $\Sigma$ of fd's, we target at rendering its relational schema for certainty (`4C,' for short). In short, it is meant to 
be the minimal-cardinality schema in BCNF that may have a lossless join.

\begin{myprob} (\textbf{Synthesis `4C'}).
Given fd set $\Sigma$, derive an fd set $\Sigma^\prime$ (causal ordering processing) to synthesize a relational schema $\bigcup_{i=1}^{n}H_k^i\!$ over it satisfying P1$^\prime$ (BCNF), P3 and striving for P4 while giving up P2$^\prime$.
\label{prob:synthesis4c}
\end{myprob}

Note in Fig. \ref{fig:pipeline} that coping with these two problems enables the loading of datasets $\bigcup_{\ell=1}^p \mathcal D_k^\ell$ into schemes \vspace{1pt}$\bigcup_{i=1}^{n}H_k^i$ to  accomplish the ETL phase of the design pipeline. The user can then benefit from hypothesis management up to the capabilities of a traditional relational DB.  
For a full-fledged tool, we shall leverage (globally) the \emph{certain} relations $\bigcup_{k=1}^{z}\bigcup_{i=1}^{n}H_k^i$ to \emph{uncertain} relations $\bigcup_{k=1}^{z}\bigcup_{j=1}^{m}Y_k^j$.

\begin{myprob} (\textbf{Synthesis `4U'}).
Given a collection \vspace{1pt}of relations $\bigcup_{i=1}^{n}H_k^i$ loaded with trial datasets $\bigcup_{\ell=1}^p \mathcal D_k^\ell$ for each hypothesis $k$, introduce properly all the uncertainty pre\-sent in $\bigcup_{k=1}^{z}\bigcup_{i=1}^{n}H_k^i$ w.r.t. encoded fd sets $\Sigma_k$ for $k\!=\!1..z$ into U-relations $\bigcup_{k=1}^{z}\bigcup_{j=1}^{m}Y_k^j$.
\label{prob:synthesis4u}
\end{myprob}

\noindent
We address problems $\!$P1-P3 in the sequel through \S\ref{sec:encoding}--\S\ref{sec:synthesis4u}.

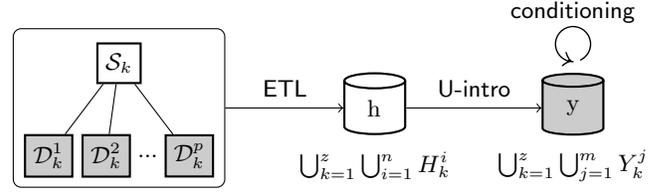
\begin{figure}[t]
\tikzstyle{rect1}=[rectangle,
                                    thick,
                                    minimum size=15pt,
                                    draw=black]
\tikzstyle{rect2}=[rectangle,
                                    thick,
                                    minimum size=15pt,
                                    fill=black!20,
                                    draw=black]
\tikzstyle{rect3}=[rectangle,
                                    rounded corners=3pt,
                                    minimum size=60pt,
                                    minimum width=80pt,
                                    draw=black]
\tikzstyle{box}=[rectangle,
                                    fill=none,
                                    draw=none]
\tikzstyle{cyl1}=[cylinder,
                                    thick,
                                    minimum size=23pt,
                                    inner sep=0pt,
                                    fill=none,
                                    draw=black]
\tikzstyle{cyl2}=[cylinder,
                                    thick,
                                    fill=black!20,
                                    minimum size=23pt,
                                    inner sep=0pt,
                                    draw=black]
\tikzstyle{edge} = [draw,thick,->,bend left]
\begin{tikzpicture}[scale=0.85]
    \node[rect3] (back) at (0,3) {};
    \node[rect1] (s) at (0,3.7) {$\mathcal{S}_k$};
    \node[rect2] (d1) at (-1.1,2.25) {$\mathcal{D}_k^1$};
    \node[rect2] (d2) at (-0.2,2.25) {$\mathcal{D}_k^2$};
    \node[box] (dot) at (0.45,2.25) {...};
    \node[rect2] (dn) at (1.1,2.25) {$\mathcal{D}_k^p$};
    \node[cyl1,rotate=90] (h) at (4,3) {\rotatebox[origin=c]{-90}{h}};
    \node[box] (hlabel) at (4,2.1) {$\bigcup_{k=1}^{z}\bigcup_{i=1}^{n}H_k^i$};
    \node[cyl2,rotate=90] (y) at (7.1,3) {\rotatebox[origin=c]{-90}{y}};
    \node[box] (ylabel) at (7.1,2.1) {$\bigcup_{k=1}^{z}\bigcup_{j=1}^{m}Y_k^j$};
    \node[box] (cond) at (7.1,4) {\huge\rotatebox[origin=c]{180}{$\circlearrowleft$}};
    \draw[-] (s) to (d1);
    \draw[-] (s) to (d2);
    \draw[-] (s) to (dn);
    \draw[->] (back) to (h);
    \node[box] (etl) at (2.6,3.3) {\textsf{ETL}};    
    \draw[->] (h) to (y);
    \node[box] (etl) at (5.55,3.3) {\textsf{U-intro}};
    \node[box] (etl) at (7.1,4.5) {\textsf{conditioning}};
\end{tikzpicture}
\vspace{-12pt}
\caption{Design-theoretic pipeline for hypothesis encoding.}
\label{fig:pipeline}
\vspace{-7pt}
\end{figure}

\begin{figure*}[t]\scriptsize
\hspace{-8pt}
\tikzset{node style ge/.style={circle,inner sep=0pt,minimum size=15pt}}
\begin{subfigure}{0.3\textwidth}
\hspace{5pt}
\begin{tikzpicture}[baseline=(A.center)]
\matrix (A) [matrix of math nodes, nodes = {node style ge},column sep=0.35 mm] {
 & \node (x1) {x_1}; & \node (x2) {x_2}; & \node (x3) {x_3}; & \node (x4) {x_4}; & \node (x5) {x_5}; & \node (x6) {x_6}; & \node (x7) {x_7};\\
\node (f1) {f_1}; & \node (a11) {1}; & \node (a12) {0}; & \node (a13) {0}; & \node (a14) {0}; & \node (a15) {0}; & \node (a16) {0}; & \node (a17) {0};\\
\node (f2) {f_2}; & \node (a21) {0}; & \node (a22) {1}; & \node (a23) {0}; & \node (a24) {0}; & \node (a25) {0}; & \node (a26) {0}; & \node (a27) {0};\\
\node (f3) {f_3}; & \node (a31) {0}; & \node (a32) {0}; & \node (a33) {1}; & \node (a34) {0}; & \node (a35) {0}; & \node (a36) {0}; & \node (a37) {0};\\
\node (f4) {f_4}; & \node (a41) {1}; & \node (a42) {1}; & \node (a43) {1}; & \node (a44) {1}; & \node (a45) {1}; & \node (a46) {0}; & \node (a47) {0};\\
\node (f5) {f_5}; & \node (a51) {1}; & \node (a52) {0}; & \node (a53) {1}; & \node (a54) {1}; & \node (a55) {1}; & \node (a56) {0}; & \node (a57) {0};\\
\node (f6) {f_6}; & \node (a61) {0}; & \node (a62) {0}; & \node (a63) {0}; & \node (a64) {1}; & \node (a65) {0}; & \node (a66) {1}; & \node (a67) {0};\\
\node (f7) {f_7}; & \node (a71) {0}; & \node (a72) {0}; & \node (a73) {0}; & \node (a74) {0}; & \node (a75) {1}; & \node (a76) {0}; & \node (a77) {1};\\
};
\end{tikzpicture}
\caption{Structure matrix as given.}\label{fig:coa-a}
\end{subfigure}
$\to$
\hspace{-2pt}
\tikzstyle{background0}=[rectangle,
                                                fill=gray!05,
                                                inner sep=0.025cm,
                                                rounded corners=1mm]
\tikzstyle{background1}=[rectangle,
                                                fill=gray!20,
                                                inner sep=0.025cm,
                                                rounded corners=1mm]
\tikzstyle{background2}=[rectangle,
                                                fill=gray!40,
                                                inner sep=0.025cm,
                                                rounded corners=1mm]
\begin{subfigure}{0.3\textwidth}
\begin{tikzpicture}[baseline=(A.center)]
  \tikzset{BarreStyle/.style =   {opacity=.4,line width=2.65 mm,line cap=round,color=#1}}
\matrix (A) [matrix of math nodes, nodes = {node style ge},column sep=0.35 mm] {
 & \node (x1) {x_1}; & \node (x2) {x_2}; & \node (x3) {x_3}; & \node (x4) {x_4}; & \node (x5) {x_5}; & \node (x6) {x_6}; & \node (x7) {x_7};\\
\node (f1) {f_1}; & \node (a11) {1}; & \node (a12) {0}; & \node (a13) {0}; & \node (a14) {0}; & \node (a15) {0}; & \node (a16) {0}; & \node (a17) {0};\\
\node (f2) {f_2}; & \node (a21) {0}; & \node (a22) {1}; & \node (a23) {0}; & \node (a24) {0}; & \node (a25) {0}; & \node (a26) {0}; & \node (a27) {0};\\
\node (f3) {f_3}; & \node (a31) {0}; & \node (a32) {0}; & \node (a33) {1}; & \node (a34) {0}; & \node (a35) {0}; & \node (a36) {0}; & \node (a37) {0};\\
\node (f4) {f_4}; & \node (a41) {1}; & \node (a42) {1}; & \node (a43) {1}; & \node (a44) {1}; & \node (a45) {1}; & \node (a46) {0}; & \node (a47) {0};\\
\node (f5) {f_5}; & \node (a51) {1}; & \node (a52) {0}; & \node (a53) {1}; & \node (a54) {1}; & \node (a55) {1}; & \node (a56) {0}; & \node (a57) {0};\\
\node (f6) {f_6}; & \node (a61) {0}; & \node (a62) {0}; & \node (a63) {0}; & \node (a64) {1}; & \node (a65) {0}; & \node (a66) {1}; & \node (a67) {0};\\
\node (f7) {f_7}; & \node (a71) {0}; & \node (a72) {0}; & \node (a73) {0}; & \node (a74) {0}; & \node (a75) {1}; & \node (a76) {0}; & \node (a77) {1};\\
};
 \draw [BarreStyle=blue] (a11.north west) to (a11.south east); 
 \draw [BarreStyle=red] (a22.north west) to (a22.south east);
 \draw [BarreStyle=green] (a33.north west) to (a33.south east);
 \draw [BarreStyle=blue] (a44.north west) to (a55.south east); 
 \draw [BarreStyle=red] (a66.north west) to (a66.south east); 
 \draw [BarreStyle=green] (a77.north west) to (a77.south east); 
     \begin{pgfonlayer}{background}
        \node [background0,
                    fit=(a11) (a12) (a13) (a14) (a15) (a16) (a17) (a21) (a22) (a23) (a24) (a25) (a26) (a27) (a31) (a32) (a33) (a34) (a35) (a36) (a37) (a41) (a42) (a43) (a44) (a45) (a46) (a47) (a51) (a52) (a53) (a54) (a55) (a56) (a57) (a61) (a62) (a63) (a64) (a65) (a66) (a67) (a71) (a72) (a73) (a74) (a75) (a76) (a77) ]
                    {};
        \node [background1,
                    fit=(a44) (a45) (a46) (a47) (a54) (a55) (a56) (a57) (a64) (a65) (a66) (a67) 
                    (a74) (a75) (a76) (a77) ]
                    {};
        \node [background2,
                    fit=(a66) (a67) (a76) (a77) ]
                    {};
    \end{pgfonlayer}
\end{tikzpicture}
\caption{COA$_t$ execution in 3 recursive steps.}\label{fig:coa-b}
\end{subfigure}
\hspace{-16pt}
\tikzstyle{rect}=[rectangle,
                                    thick,
                                    minimum size=0.3cm,
                                    draw=black]
\tikzstyle{circ}=[circle,
                                    thick,
                                    minimum size=0.3cm,
                                    draw=black]
\tikzstyle{vertex}=[circle,fill=black!10,minimum size=20pt,inner sep=0pt]
\tikzstyle{selected vertex} = [vertex, fill=red!24]
\tikzstyle{edge} = [draw,thick,->,bend left]
\tikzstyle{weight} = [font=\small]
\tikzstyle{selected edge} = [draw,line width=5pt,-,red!50]
\tikzstyle{ignored edge} = [draw,line width=5pt,-,black!20]
\begin{subfigure}{0.3\textwidth}
\vspace{12pt}
\hspace{8pt}
\begin{tikzpicture}[scale=0.85]
    \foreach \pos/\name in {{(-0.5,2)/x_1}, {(2,2)/x_2}, {(4.5,2)/x_3}, 
                            		 {(0.7,0)/x_4}, {(3.3,0)/x_5},
		 			 {(1,-2)/x_6}, {(3,-2)/x_7}}
        \node[vertex] (\name) at \pos {$\name$};
    \draw[->] (x_1) to (x_4);
    \draw[->] (x_1) to (x_5);
    \draw[->] (x_2) to (x_4);
    \draw[->] (x_3) to (x_4);
    \draw[->] (x_3) to (x_5);
    \draw[->] (x_4) to[out=10,in=200] (x_5);
    \draw[->] (x_5) to[out=170,in=-20] (x_4);
    \draw[->] (x_4) to (x_6);
    \draw[->] (x_5) to (x_7);
\end{tikzpicture}
\vspace{-5pt}
\caption{Directed causal graph $G_{\varphi_t}$.}\label{fig:coa-c}
\end{subfigure}
\caption{Running Simon's Causal Ordering Algorithm (\textsf{COA}) on a given structure matrix (Fig. \ref{fig:coa-a}). Minimal subsets detected in each recursive step (highlighted in different shades of gray) have their diagonal elements colored (Fig. \ref{fig:coa-b}).}
\label{fig:coa}
\vspace{-6pt}
\end{figure*}
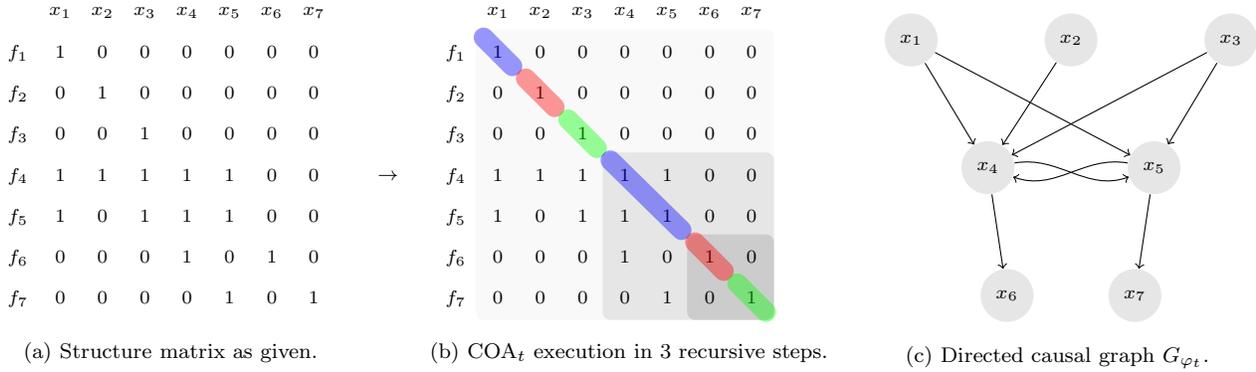

\section{Hypothesis Encoding}\label{sec:encoding}
\vspace{-3pt}
\noindent
In this section we present a technique to address Problem \ref{prob:h-encode}. 
For the encoding we shall consider a set $Z$ of attribute symbols such that $Z \!\simeq\! \mathcal V$, where $\mathcal S(\mathcal E, \mathcal V)$ is a complete structure; and two special attribute symbols, $\phi,\, \upsilon \notin Z$, which are kept to identify (resp.) phenomena and hypotheses. We are explicitly distinguishing symbols in $Z$, assigned by the user into structure $\mathcal S$, from epistemological symbols $\phi$ and $\upsilon$. 
Now, we consider a sense of Simon's into the nature of scientific modeling and interventions \cite{simon1953}, summarized in Def. \ref{def:exogenous}.

\begin{mydef}
Let $\mathcal S(\mathcal E, \mathcal V)$ be a structure and $x_\ell \in \mathcal V$ be a variable. We say that $x_\ell$ is \textbf{exogenous} if there exists an equation $f_k \!\in \mathcal E$ that can be written $f_k(x_\ell) \!= 0$, i.e., $A_S(k, j) \!=\! 1$ iff $j = \ell$. 
We say that $x_\ell$ is \textbf{endogenous} otherwise.
\label{def:exogenous}
\end{mydef}

\vspace{-3pt}
Remark \ref{rmk:modeling} introduces an interpretation of Def. $\!$\ref{def:exogenous} with a data dependency flavor.

\begin{myremark}
The value of exogenous variables (attri\-butes) is determined empirically, outside of the system (proposed structure $\mathcal S$). Such values are, therefore, dependent on the phenomenon id $\phi$ only. The value of endogenous variables (attributes) is in turn determined theoretically, within the system. They are dependent on the hypothesis id $\upsilon$ and shall be dependent on the phenomenon id $\phi$ as well. 
$\Box$
\label{rmk:modeling}
\end{myremark}

\noindent
We give (Alg. \ref{alg:coa}) \textsf{COA}$_t$, which is a (more detailed) variant of Simon $\!$(and$\!$ Dash-Druzdzel)'s $\!$\textsf{COA}. It returns a total causal mapping $\varphi_t$, instead of a partial causal mapping. We illustrate it through Example \ref{ex:struct-matrix} and Fig. \ref{fig:coa}.

\begin{algorithm}[h]\footnotesize
\caption{COA$_t$ as a variant of Simon's COA.}
\label{alg:coa}
\begin{algorithmic}[1]
\Procedure{\textsf{COA}$_t$}{${\mathcal S\!:\, \text{structure over}\; \mathcal E \;\text{and}\; \mathcal V}$}
\Require $\mathcal S$ given is complete, i.e., $|\mathcal E|=|\mathcal V|$
\Ensure Returns total causal mapping $\varphi_t: \mathcal E \to \mathcal V$
\State $\varphi_t \gets \varnothing$, $\mathcal{S}_c \gets \varnothing$ 
\ForAll{minimal $\mathcal S^\prime \subset \mathcal S$}\vspace{1pt}
\State $\mathcal{S}_c \gets \mathcal{S}_c \cup \mathcal S^\prime$ \Comment{store minimal structures in $\mathcal S$}\vspace{1pt}
\State $\mathcal V^\prime \gets \mathcal S^\prime(\mathcal V)$
\ForAll{$f \in \mathcal S^\prime(\mathcal E)$}
\State $x \gets \text{any} \;x_a \in \mathcal V^\prime$\vspace{1pt}
\State $\varphi_t \gets \varphi_t \cup \langle f,\, x \rangle$\vspace{1pt}
\State $\mathcal V^\prime \gets \mathcal V^\prime \setminus \{x\}$
\EndFor
\EndFor
\State $\mathcal T \gets \mathcal S \setminus \bigcup_{\mathcal S^\prime \in \mathcal{S}_c} \mathcal S^\prime$\vspace{1pt}
\If{$\mathcal T \neq \varnothing$}
\State \Return $\varphi_t \;\cup\;$\textsf{COA}$_t(\mathcal T)$
\EndIf
\State \Return $\varphi_t$
\EndProcedure
\end{algorithmic}
\end{algorithm}

\vspace{-5pt}
\begin{myex}
Consider structure $\mathcal S(\mathcal E, \mathcal V)$ whose matrix is shown in Fig. $\!$\ref{fig:coa-a}. $\!$Note that $\mathcal S$ is complete, since $|\mathcal E|\!=\!|\mathcal V|\!=\!7$, but not minimal. The set of all minimal subsets $\mathcal S^\prime\! \subset \mathcal S$ is $\mathcal{S}_c \!=\! \{\, \{f_1\},\, \{f_2\},\, \{f_3\} \,\}$. $\!$By eliminating the variables identified at recursive step $k$, a smaller structure $\mathcal T \subset \mathcal S$ is derived. Compare the partial causal mapping eventually returned by $\!$\textsf{COA}, $\varphi_p \supset\! \{\, \langle\{f_4, f_5\},\, \{x_4, x_5\}\rangle \,\}$, to the total causal mapping returned by \textsf{COA}$_t$, $\,\varphi_t \supset \{\, \langle f_4, x_4 \rangle,\, \langle f_5, x_5\rangle \,\}$. $\!$Since $x_4$ and $x_5$ are strongly coupled (see Fig.\ref{fig:coa-b}), \textsf{COA}$_t$ maps them arbitrarily (i.e., it could be $f_4 \mapsto x_5,\, f_5 \mapsto x_4$ instead). Such total mapping $\varphi_t$ renders a cycle in the directed causal graph $G_{\varphi_t}$ (see Fig.\ref{fig:coa-c}). $\Box$
\label{ex:struct-matrix}
\end{myex}

\noindent
We encode complete structures into fd sets by means of (Alg. $\!$\ref{alg:h-encode}) \textsf{h-encode}. $\!$Fig. \ref{fig:h-fdschema} \mbox{(left) presents an fd set defined} $\Sigma \triangleq$ \textsf{h-encode}($\mathcal S$), where $\mathcal S$ is shown in Fig. $\!$\ref{fig:coa}. Next we study the main properties of the encoded fd sets.

\begin{algorithm}[H]\footnotesize
\caption{Hypothesis encoding.}
\label{alg:h-encode}
\begin{algorithmic}[1]
\Procedure{h-encode}{$\mathcal S\!:\, \text{structure over}\; \mathcal E \;\text{and}\; \mathcal V$}
\Require $\mathcal S$ given is a complete structure, i.e., $|\mathcal E|=|\mathcal V|$
\Ensure Returns a non-redundant fd set $\Sigma$
\State $\Sigma \gets \varnothing$
\State $\varphi_t \gets \textsf{COA}_t(\mathcal S)$\vspace{1pt}
\ForAll{$\langle f_k,\, x_\ell \rangle \in \varphi_t$}\vspace{1pt}
\State $Z \gets x_j \,\;\text{for all}\;\, j \,\;\text{such that}\; A_S(k, j)=1$\vspace{1pt}
\If{$|Z|=1$} \Comment{$x_\ell$ is exogenous}\vspace{1pt}
\State $\Sigma \gets \Sigma \cup \langle \{\phi\}, \{x_\ell\} \rangle$
\Else \Comment{$x_\ell$ is endogenous}\vspace{1pt}
\State $\Sigma \gets \Sigma \cup \langle Z\!\setminus\! \{x_\ell\} \cup \{\upsilon\},\; \{x_\ell\} \rangle$
\EndIf
\EndFor
\State \Return $\Sigma$
\EndProcedure
\end{algorithmic}
\end{algorithm}

\begin{mylemma}
$\!$Let $\,\Sigma$ be a singleton-rhs fd set on attributes $U$. Then $\langle X, A\rangle \in \Sigma^+$ with $XA \subseteq U$ only if $A \subseteq X$ or there is non-trivial $\langle Y, A\rangle \in \Sigma$ for some $Y \subset U$.
\label{lemma:singleton-rhs}
\end{mylemma}
\begin{proof}
See Appendix, \S\ref{a:singleton-rhs}.
\end{proof}

\begin{mythm}
$\!\!$Let $\Sigma$ be an fd set defined $\Sigma \!\triangleq\!$ \textsf{h-encode}($\mathcal S$) for some complete structure $\mathcal S$. Then it is non-redundant but may not be canonical. 
\label{thm:non-redundant}
\end{mythm}
\begin{proof}
We show that properties (a-b) of Def. \ref{def:minimal} must hold for $\Sigma$ produced by (Alg. \ref{alg:h-encode}) \textsf{h-encode}, but property (c) may not hold (i.e., encoded fd set $\Sigma$ may not be left-reduced). See Appendix, \S\ref{a:non-redundant}.
\end{proof}

Finally, it shall be convenient to come with a notion of \emph{parsimonious} fd sets (see Def. $\!$\ref{def:parsimonious}). This is bit stronger an assumption than canonical fd sets, yet provably just fit to our use case (cf. Corollary $\!$\ref{cor:parsimonious} and its proof in \S\ref{a:parsimonious}).

\begin{mydef}
Let $\Sigma$ be set of fd's on attributes $U$. Then, we say that $\Sigma$ is \textbf{parsimonious} if it is canonical and, for all fd's $\langle X, A\rangle \in \Sigma$ with $XA \subseteq U$, there is no $Y \subset U$ such that $Y \neq X$ and $\langle Y, A\rangle \in \Sigma$.
\label{def:parsimonious}
\end{mydef}

\begin{mycor}
Let $\Sigma$ be an fd set defined $\Sigma \!\triangleq\!$ \textsf{h-encode}($\mathcal S$) for some complete structure $\mathcal S$. Then $\Sigma$ can be assumed parsimonious with no loss of generality at expense of time that is polynomial in $|\Sigma|\cdot|U|$.
\label{cor:parsimonious}
\end{mycor}
\vspace{-9pt}
\begin{proof}
See Appendix, \S\ref{a:parsimonious}.
\end{proof}
\vspace{-3pt}
We draw attention to the significance of Theorem \ref{thm:non-redundant}, as it sheds light on a connection between Simon's complete structures \cite{simon1953} and fd sets \cite{ullman1988}. In fact, we continue to elaborate on that connection in next section to ensure the synthesis of relational schemas with properties of our interest. We may assume given fd sets to be canonical or parsimonious if they are defined by \textsf{h-encode} (Corollary \ref{cor:parsimonious}).

\begin{figure}[t]
\begin{framed}
\vspace{-11pt}
\begin{subfigure}{0.48\columnwidth}
\begin{eqnarray*}
\Sigma = \{\quad 
\phi &\to& x_1,\\ 
\phi &\to& x_2,\\ 
\phi &\to& x_3,\\ 
x_1\,x_2\,x_3\,x_5\,\upsilon &\to& x_4,\\
x_1\,x_3\,x_4\,\upsilon &\to& x_5,\\
x_4\,\upsilon &\to& x_6,\\
x_5\,\upsilon &\to& x_7 \quad\}.
\end{eqnarray*}
\end{subfigure}
\hspace{3pt}
\begin{subfigure}{0.48\columnwidth}
\begin{eqnarray*}
\Sigma^\looparrowright = \{\quad 
\phi &\to& x_1,\\ 
\phi &\to& x_2,\\ 
\phi &\to& x_3,\\ 
\phi\,\upsilon\;x_5 &\to& x_4,\\
\phi\,\upsilon\;x_4 &\to& x_5,\\
\phi\,\upsilon\;x_5 &\to& x_6,\\
\phi\,\upsilon\;x_4 &\to& x_7 \quad\}.
\end{eqnarray*}
\end{subfigure}
\end{framed}
\vspace{-11pt}
\caption{Primitive fd set $\Sigma$ encoding (cf. $\!$Alg. $\!$\ref{alg:h-encode}) the structure of Fig. \ref{fig:coa-a} and its folding $\Sigma^\looparrowright$ derived by Alg. $\!$\ref{alg:folding}.}
\label{fig:h-fdschema}
\vspace{-9pt}
\end{figure}

\vspace{-2pt}
\section{Synthesis `4C'}\label{sec:synthesis4c}
\vspace{-2pt}
\noindent
In this section we present a technique to address Problem \ref{prob:synthesis4c}. 
Recall that we aim at a synthesis method to ensure the produced schema $\boldsymbol R$ bears some desirable properties, viz., P1$^\prime$ (BCNF), P3 and P4; give up P2$^\prime$.  
Let us then consider procedure \textsf{synthesize} (Alg. \ref{alg:synthesize}). 
This algorithm is essentially Bernstein's \cite{bernstein1976}. 
In our use\vspace{-2pt} case, input fd sets are defined $\Sigma\! \triangleq$ \textsf{h-encode}($\mathcal S$) and then may safely assumed to be parsimonious. Rather than modifying that classical algorithm, we shall achieve the properties we want for a synthesized schema by a very specific manipulation on input fd set $\Sigma$.

\begin{algorithm}[H]\footnotesize
\caption{Schema synthesis.}
\label{alg:synthesize}
\begin{algorithmic}[1]
\Procedure{synthesize}{$\Sigma: \text{fd set}$}
\Require $\Sigma$ given is parsimonious, and let $U\!:=\!Attrs(\Sigma)$
\Ensure Returns schema $\boldsymbol R[U]$ in 3NF that preserves $\Sigma$
\State $\Sigma^\prime \gets\; \text{apply (R4) union to} \;\Sigma$ 
\State $\boldsymbol R \gets \varnothing$ 
\ForAll{$\langle X,\, Z \rangle \in \Sigma^\prime$}
\If{there is $R_k[YW] \in \boldsymbol R$ such that $X \leftrightarrow Y\!$}
\State $R_k \gets R_k \cup XZ$ 
\Else
\State $R_{i+1} \gets XZ$, with designated key $X$
\State $\boldsymbol R \gets \boldsymbol R \cup R_{i+1}$
\EndIf
\EndFor
\State \Return $\boldsymbol R$
\EndProcedure
\end{algorithmic}
\end{algorithm}

\begin{myprop}
Let $\boldsymbol R[U]$ be a relational schema,\vspace{-2pt} defined $\boldsymbol R \triangleq$ \textsf{synthesize}($\Sigma$) for some  canonical fd set $\Sigma$ on attributes $U$. Then $\boldsymbol R$ preserves $\Sigma$, and is in 3NF but may not be in BCNF. 
\label{prop:3nf}
\end{myprop}
\vspace{-10pt}
\begin{proof}
See Appendix, \S\ref{a:3nf}.
\end{proof}

\begin{myremark}
\label{rmk:3nf}
$\!$The connection of Proposition \ref{prop:3nf} with Theorem \ref{thm:non-redundant} and Corollary \ref{cor:parsimonious} reveals an interesting result, viz., the encoding of complete structures (i.e., derived from determinate systems of mathematical equations) followed by straightforward synthesis always leads to 3NF relational schemas. It is suggestive of the precision and (verifiable) consistency of mathematical systems, in comparison to arbitrary information systems. $\Box$
\vspace{-3pt}
\end{myremark}
In Proposition \ref{prop:3nf} we were not concerned with the recoverability of data. Classical versions of (Alg. \ref{alg:synthesize}) \textsf{synthesize} include an (artificial) additional step to ensure the lossless join property (cf. $\!$\cite[p. $\!$257-8]{abiteboul1995}). We post\-pone the study of Alg. \ref{alg:synthesize} w.r.t. that property to \S\ref{subsec:synthesis4c}.

So far, we know that any relational schema synthesized straightforwardly from its primitive fd set $\Sigma\! :=$ \textsf{h-encode}($\mathcal S$) is in 3NF and is dependency-preserving. Yet, it does not give us a ``good'' design in the sense of Problem \ref{prob:synthesis4c} (cf. \S\ref{subsec:problem}). It fails w.r.t.$\!$ P1$^\prime$ (BCNF) and P3 (minimal-cardinality schema). For example, \textsf{synthesize} over $\Sigma$ given in Fig. \ref{fig:h-fdschema} (left) produces $|\boldsymbol R|\!=\!5$, while we target at a more succinct, less decomposed schema. We shall give up strict dependency preservation to go beyond 3NF towards BCNF for a more compact representation of the causal ordering `embedded' in $\Sigma$ (cf. $\!$\S\ref{subsec:ptc}). 
For hypothesis management, a less redundant schema (BCNF over 3NF) matters not because of update anomalies but succintness. Interestingly, Arenas and Libkin have shown in information-theoretic terms how `non-redundant' schemes in BCNF are \cite{arenas2005}.

\vspace{-2pt}
\subsection{Reflexive Pseudo-Transitive Reasoning}\label{subsec:ptc}
\vspace{-3pt}
\noindent
We seek the most succinct schema that somewhat preserves the causal ordering of the fd set given. That is achievable (see e.g., $\Sigma^\looparrowright\!$ in Fig. \ref{fig:h-fdschema}, right) by reflexive pseudo-transitive reasoning over $\Sigma$.

\vspace{-2pt}
\begin{mydef}
Let $\Sigma$ be a set of fd's on attributes $U\!$. Then $\Sigma^{\triangleright}$, the \textbf{reflexive pseudo-transitive closure} of $\Sigma$, is the set $\Sigma^{\triangleright} \supseteq \Sigma$ such that $X \!\to Y$ is in $\Sigma^{\vartriangleright}\!$, with $XY\! \subseteq U$, iff it can be derived from a finite (possibly empty) application of rules R0, R5 over fd's in $\Sigma$. In\vspace{-2pt} that case, we may write $X \!\xrightarrow{\triangleright} Y$ and omit `w.r.t. $\!\Sigma$' if it can be understood from the context.
\label{def:ptc}
\end{mydef}

We are in fact interested in a very specific proper subset of $\Sigma^\vartriangleright\!$, say, a kernel of fd's in $\Sigma^\vartriangleright$ that gives a ``compact'' representation of the causal ordering `embedded' in $\Sigma$. Note that, to characterize such special subset we shall need to be careful w.r.t. the presence of cycles in the causal ordering.

\begin{framed}
\vspace{-13pt}
\begin{mydef}
$\!$Let $\Sigma$ be a set of fd's on attributes $U\!$, and $\langle X, \!A\rangle \in \Sigma^\vartriangleright\!$ with $X\!A \subseteq U$. $\!$We say that $X \!\to\! A$ is \textbf{folded} (w.r.t.$\!$ $\Sigma$), and write $X \xrightarrow{\looparrowright} A$, if it is non-trivial and for no $Y \subset U$ with $Y \!\nsupseteq X$, we have $Y \!\to X$ and $X \!\not\to Y$ in $\Sigma^+$.
\label{def:folded}
\end{mydef}
\vspace{-13pt}
\end{framed}

The intuition of Def. \ref{def:folded} is that an fd is folded when there is no sense in going on with pseudo-transitive reasoning over it anymore. Given an fd $X \!\to A$ in fd set $\Sigma$, we shall be able to find some folded fd $Z \!\to A$ by applying (R5) pseudo-transitivity as much as possible while ruling out cyclic or trivial fd's in some clever way.

\begin{mydef}
$\!$Let $\Sigma$ be an fd set on attributes $U$, and $\langle X, A\rangle \in \Sigma$ be an fd with $XA \subseteq U$. Then,
\begin{itemize}
\item[(a)] $A^\looparrowright\!$, the \textbf{(attribute) folding} of $A$ (w.r.t. $\!\!\Sigma$) is an\vspace{-2pt} attribute set $Z \!\subset U$ such that $Z \xrightarrow{\looparrowright} A$;

\item[(b)] Accordingly, $\Sigma^\looparrowright\!$, the \textbf{folding} of $\Sigma$, is a proper subset $\Sigma^\looparrowright\! \subset \Sigma^\vartriangleright$ such that an fd $\langle Z, A\rangle \in \Sigma^\vartriangleright\!$ is in $\Sigma^\looparrowright$\vspace{-2pt} iff $X \xrightarrow{\looparrowright} A$ for some $Z \subset U$. 
\end{itemize}
\label{def:folding}
\end{mydef}

\vspace{-8pt}
\setcounter{myex}{0}
\begin{myex}$\!$(continued).$\!$ 
Fig. $\!\!$\ref{fig:h-fdschema} shows an fd set $\Sigma$ (left) and its folding $\Sigma^\looparrowright\!\!$ (right). $\!$Note that the folding can be obtai\-ned by computing the attribute folding for $A$ in each fd $X \!\to A$ in $\Sigma$. We illustrate below some reasoning steps to partially compute an attribute folding. 
\begin{eqnarray*}
1.\,\;\;\;\;\;\;\;\;\;\;\;\phi\,\upsilon\,x_4 &\to& x_5 \quad \text{\emph{[consider given]}}\\
2.\,\;\;\;\;\;\;\;\;\;\;\;\phi\,\upsilon\,x_5 &\to& x_4 \quad \text{\emph{[consider given]}}\\
3.\,\;\;\;\;\;\;\;\;\;\;\;\;\;\;x_4\,\upsilon &\to& x_6 \quad \text{\emph{[given]}}\\
\line(1,0){70}&\!\!\!\!\!\! \line(1,0){30}\!\!\!\!\!\!&\line(1,0){93}\\\vspace{-2pt}
4.\;\;\therefore\,\;\;\;\;\;\phi\,\upsilon\, x_5 &\to& x_6 \quad \text{\emph{[R5 over (2), (3)]}}.
\vspace{4pt}
\end{eqnarray*}
Note that (4) is still amenable to further application of R5, say, over (1), (4), to derive (5) $\phi\,\upsilon\,x_4 \!\to x_6$. However, even though (4) and (5) have (resp.) the form $X \!\to A$ and $Y \!\to A$ with $Y \!\to X$, we have $X \!\to Y$ as well which characterizes a cycle. In fact, (4) itself satisfies Def. \ref{def:folded} and then is folded (w.r.t. $\!\Sigma$ from Fig. \ref{fig:h-fdschema}). The same holds for (1) and (2). 
$\Box$ 
\end{myex}

\begin{mylemma}
Let $\Sigma$ be a parsimonious fd set on attributes $U$, and $\langle X, A\rangle \!\in\! \Sigma$ be an fd with $XA \subseteq U$. Then $A^\looparrowright\!$, the attribute folding of $A$ (w.r.t. $\!\Sigma$) exists. Moreover, if $\Sigma$ is parsimonious then $A^\looparrowright\!$ is unique.
\label{lemma:folding-unique}
\end{mylemma}
\begin{proof}
See Appendix, \S\ref{a:folding-unique}.
\end{proof}

We give an original algorithm (Alg. $\!$\ref{alg:folding}) to com\-pute the folding of an fd set. At its core there lies (Alg. $\!$\ref{alg:afolding}) \textsf{AFolding}, which can be understood as a non-obvious variant of \textsf{XClosure} (cf. $\!$Alg. $\!$\ref{alg:xclosure}) designed for acyclic reflexive pseudo-transitivity reasoning. 
In order to compute the folding of attribute $A$ in fd $\langle X, A\rangle \in \Sigma$, algorithm \textsf{AFolding} backtraces the causal ordering `embedded' in $\Sigma$ towards $A$. Analogously, in terms of the directed graph $G_{\varphi_t}$ induced by the causal ordering (see Fig. \ref{fig:coa-c}), that would comprise graph traversal to identify the nodes $x_p$ that have $x \mapsto A$ in their reachability, $x_p \rightsquigarrow x$. Rather, \textsf{AFolding}'s processing of the causal ordering is fully symbolic based on Armstrong's rewrite rules R0, R5.

\begin{myex}
Cyclicity in an fd set $\Sigma$ may have the effect of making its folding $\Sigma^\looparrowright\!$ to degenerate to $\Sigma$ itself. For instance, consider $\Sigma\!=\!\{A \!\to B,\, B \!\to A\}$. Note that $\Sigma$ is parsimonious, and \textsf{AFolding} (w.r.t. $\Sigma$) is $B$ given $A$, and $A$ given $B$. That is, $\Sigma^\looparrowright \!= \Sigma$. 
$\Box$
\label{ex:degenerate}
\end{myex}

\begin{algorithm}[t]\footnotesize
\caption{Folding of an fd set.}
\label{alg:folding}
\begin{algorithmic}[1]
\Procedure{folding}{$\Sigma\!: \text{fd set}$} 
\Require $\Sigma$ given is parsimonious
\Ensure Returns fd set $\Gamma=\Sigma^\looparrowright$, the folding of $\Sigma$
\State $\Gamma \gets \varnothing$

\ForAll{$\langle X,\, A \rangle \in \Sigma$}
\State $Z \gets\, \textsf{AFolding}(\Sigma,\, A)$
\State $\Gamma \gets \Gamma \cup \langle Z,\, A \rangle$
\EndFor
\State \Return $\Gamma$
\EndProcedure
\end{algorithmic}
\end{algorithm}
%
\begin{algorithm}[h]\footnotesize
\caption{Folding of an attribute w.r.t. an fd set.}
\label{alg:afolding}
\begin{algorithmic}[1]
\Procedure{AFolding}{$\Sigma\!: \text{fd set},\; A\!: \text{attribute}$} 
\Require $\Sigma$ is parsimonious
\Ensure Returns $A^\looparrowright\!$, the attribute folding of $A$ (w.r.t. $\Sigma$)
\State $\Lambda \gets \varnothing$  \Comment{consumed attrs.}
\State $\Delta \gets \varnothing$ \Comment{consumed fd's}
\State $A^\star \gets A$ \Comment{store ``causal parent'' attrs. of $A$}
\State \emph{size} $\gets 0$
\While{\emph{size} $< |A^\star|$ } \Comment{halt when $A^{(i+1)}\!=\!A^{(i)}$}
\State \emph{size} $\gets |A^\star|$
\State $\Sigma \gets \Sigma \setminus \Delta$
\ForAll{$\langle Y,\, B \rangle \in \Sigma$}
\If{$B \in A^\star$}
\State $\Delta \gets \Delta \cup \{\langle Y,\, B \rangle\}$ \Comment{consume fd}
\State $A^\star \gets\, A^\star\! \cup Y$
\If{$Y \!\cap \Lambda=\varnothing$} \Comment{non-cyclic fd}
\State $\Lambda \gets\, \Lambda \cup B$ \Comment{consume attr.}
\EndIf
\EndIf
\EndFor
\EndWhile
\State \Return $A^\star \setminus \Lambda$
\EndProcedure
\end{algorithmic}
\end{algorithm}

\begin{mythm}
Let $\Sigma$ be a parsimonious fd set on attributes $U$, and $A$ be an attribute with $\langle X, A\rangle \in \Sigma$ with $XA \subseteq U$. Then \textsf{AFolding}($\Sigma, A$) correctly computes $A^\looparrowright\!$, the attribute folding of $A$ (w.r.t. $\!\Sigma$) in time $O(n^2)$ in $|\Sigma| \cdot |U|$.
\label{thm:afolding}
\end{mythm}
\begin{proof}
For the proof roadmap, note that \textsf{AFolding} 
is monotone and terminates precisely when $A^{(i+1)}\!=\!A^{(i)}$, where $A^{(i)}$ denotes the attributes in $A^\star$ at step $i$ of the outer loop. The folding $A^\looparrowright\!$ of $A$ at step $i$ is $A^{(i)} \setminus \Lambda^{(i)}$. We shall prove by induction, given attribute $A$ in fd $X \!\to A$ in parsimonious $\Sigma$, that $A^\star \!\setminus \Lambda$ returned by \textsf{AFolding}($\Sigma, A$) is the unique attribute folding $A^\looparrowright\!$ of $A$. See Appendix, \S\ref{a:afolding}.
\end{proof}

\begin{myremark}
\label{rmk:linear}
Beeri and Bernstein gave a straightforward optimization to (Alg. \ref{alg:xclosure}) \textsf{XClosure} to make it linear in $|\Sigma| \!\cdot\! |U|$ (cf. $\!$\cite[p. $\!\!$43-5]{beeri1979}). It applies likewise to (Alg. \ref{alg:afolding}) \textsf{AFolding}, but we omit its tedious exposure here and simply consider that \textsf{AFolding} can be implemented to be $O(n)$ in $|\Sigma| \cdot |U|$.\footnote{In short, it shall require one more auxiliary data structure to keep track, for each fd not yet consumed, of how many attributes not yet consumed appear in its rhs.} $\Box$
\end{myremark}

\begin{mycor}
Let $\Sigma$ be a canonical fd set on attributes $U$. Then algorithm \textsf{folding}($\Sigma$) correctly computes $\Sigma^\looparrowright\!$, the folding of $\Sigma$ in time that is $f(n)\,\Theta(n)$ in the size $|\Sigma| \cdot |U|$, where $f(n)$ is the time complexity of (Alg. \ref{alg:afolding}) \textsf{AFolding}.
\label{cor:folding}
\end{mycor}
\begin{proof}
See Appendix, \S\ref{a:folding}.
\end{proof}

Finally, another property of the folding of an fd set which shall be useful to know is given by Proposition \ref{prop:folding-parsimonious}.

\begin{myprop}
Let $\Sigma$ be an fd set, and $\Sigma^\looparrowright\!$ its folding. If $\Sigma$ is parsimonious then so is $\Sigma^\looparrowright\!$.
\label{prop:folding-parsimonious}
\end{myprop}
\begin{proof}
See Appendix, \S\ref{a:folding-parsimonious}.
\end{proof}

\subsection{Schema Synthesis over the Folding \large{$\,\Sigma^\looparrowright$}}\label{subsec:synthesis4c}
\vspace{-2pt}
\noindent
We motivate our goal of computing the folding to carry out schema synthesis over it by means of Example \ref{ex:folding}. 

\begin{myex}
Let us consider canonical fd set $\Sigma\!=\!\{{A \!\to B,}$ $B \!\to C\,\}$ over attributes $U\!=\!\{A, B, C\}$, and a tentative schema containing a single relation $R[A B C]$. This relation is not in BCNF because, for one, $B \!\to C$ violates it ($C \nsubseteq B$ but $B$ is not a superkey for $R$). A typical approach to provide a BCNF schema is to apply a `decomposition into BCNF' algorithm (cf. $\!$\cite{abiteboul1995}) to get BCNF schema $\boldsymbol R \!=\! R_1[A B] \,\cup\, R_2[B C]$. Instead, it suffices for us to consider the folding $\Sigma^\looparrowright\!=\{A \!\to B,\, A \!\to C\}$ of $\,\Sigma$. By straightforward synthesis, we generate $R[A B C]$ which is BCNF w.r.t. $\!\Sigma^\looparrowright$. 
$\Box$
\label{ex:folding}
\end{myex}

\vspace{-2pt}
Schema synthesis over $\!\Sigma^\looparrowright\!$, the folding of a parsimonious fd set $\Sigma\!$, gives up preservation of $\Sigma$ to target at a BCNF schema that somewhat preserves the causal ordering `embedded' in $\Sigma$, i.e., preserves $\!\Sigma^\looparrowright\!$.
Now we review the properties of relational schema $\boldsymbol R$ as then synthesized over the folding $\!\Sigma^\looparrowright\!$ of $\Sigma$. Recall from Proposition \ref{prop:3nf} that synthesis over $\!\Sigma$ may render a schema $\boldsymbol R$ not in BCNF. The problem of deciding whether a given $\boldsymbol R$ is in BCNF is NP-complete 
\cite[p. $\!$256]{abiteboul1995}.  
However, by Theorem \ref{thm:bcnf} we shall guarantee the BCNF property a priori for every schema synthesized over the folding.

\begin{mythm}
Let $\boldsymbol R[U]$ be a relational schema, defined $\boldsymbol R \triangleq$ \textsf{synthesize}($\Sigma^\looparrowright\!$), where $\Sigma^\looparrowright\!$ is the folding of parsimonious fd set $\Sigma$ on attributes $U$. We claim that $\boldsymbol R$ is in BCNF, is minimal-cardinality and preserves $\Sigma^\looparrowright\!$.
\label{thm:bcnf}
\end{mythm}
\vspace{-7pt}
\begin{proof}
See Appendix, \S\ref{a:bcnf}.
\end{proof}

\begin{myprop}
\label{prop:lossless}
Let $\boldsymbol R[U] \triangleq$ \textsf{synthesize}($\Sigma^\looparrowright\!$) be a relational schema with $|\boldsymbol R|\geq 2$, where $\Sigma^\looparrowright\!$ is the folding of a parsimonious fd set $\Sigma \triangleq$ \textsf{h-encode}($\mathcal S$) on attributes $U$. Then $\boldsymbol R$ has a lossless join (w.r.t. $\Sigma^\looparrowright\!$) iff, for all $R_i[XZ] \in \boldsymbol R$ with key constraint $X \to Z$, we have $X \to U$ or there is $R_j[YW] \in \boldsymbol R$ such that $X \subset Y$.
\end{myprop}
\begin{proof}
See Appendix, \S\ref{a:lossless}.
\end{proof}

\begin{myremark}
An alternative approach (cf. $\!$\cite[p. $\!$411]{ullman1988}) to ensure the lossless join property w.r.t. $\!\Sigma^\looparrowright\!$ over attributes $U$ is to render an additional ``artificial'' scheme $R_{i+1}[X]$, where $X$ is any superkey for $U$, in order to get $\boldsymbol R^\prime := \boldsymbol R \cup R_{i+1}[X]$. Such $\boldsymbol R^\prime$ is in BCNF and has a lossless join for sure but is not the minimal-cardinality schema in BCNF and then is not considered here. $\Box$
\end{myremark}

\begin{myex}
Apply $\boldsymbol R \triangleq$ \textsf{synthesize}($\Sigma^\looparrowright\!$), where $\Sigma^\looparrowright\!$ is given in Fig. \!\ref{fig:h-fdschema} (right). \!Then we get $\boldsymbol R \!=\! \{R_1[\phi\, x_1\, x_2\, x_3],$ $R_2[\phi\, \upsilon\, x_5\, x_4\, x_6\, x_7]\,\}$, which is in BCNF, preserves $\Sigma^\looparrowright\!$ and has a lossless join. Now, let us take a slightly different fd set $\Gamma \triangleq \Sigma \cup \{x_1\,x_9\,\upsilon \!\to\! x_8,\; x_2\,x_8\,\upsilon \!\to\! x_9\}$. By applying $\boldsymbol R^\prime \triangleq$ \textsf{synthesize}($\Gamma^\looparrowright\!$), we get $\boldsymbol R^\prime = \boldsymbol R \cup \{R_3[\phi\, \upsilon\, x_8, x_9]\}$, which is in BCNF, preserves $\Sigma^\looparrowright\!$ but does not have a lossless join. It turns out that $\Gamma^\looparrowright\!$ ``embeds'' two subsets of strongly coupled variables (attributes), viz. $\!\{x_4, x_5, x_6, x_7\}$ and $\{x_8, x_9\}$ that are not ``causally connected'' to each other. 
$\Box$
\label{ex:lossless}
\end{myex}

\begin{myconj}
The lossless join property is reducible to the structure $\mathcal S$ given as input to the pipeline. 
\label{conj:lossless}
\end{myconj}
\noindent
We comment on Conjecture \ref{conj:lossless} in some detail in \S\ref{a:conjecture}.
In the converse direction, we bring in Def. \ref{def:exo-endo} SEM's concepts into data dependency language.

\begin{mydef}
Let $H[XZ]$ be a relation with key constraint $X \!\to Z$. We say that $X \!\to Z$ is a \textbf{$\boldsymbol \phi$-fd} over \textbf{exogenous} attributes $Z$ (and exogenous relation $H[XZ]$) if $\upsilon \notin X$. We say that it is an \textbf{$\boldsymbol \upsilon$-fd} over \textbf{endogenous} attributes $Z$ (and endogenous relation $H[XZ]$) otherwise.
\label{def:exo-endo}
\end{mydef}

\section{Synthesis `4U'}\label{sec:synthesis4u}
\vspace{-2pt}
\noindent
In this section we present a technique to address Problem \ref{prob:synthesis4u}. At this stage of the pipeline, relational schema $\boldsymbol H$ has been synthesized and datasets computed from the hypotheses under alternative trials (input settings) are loaded into it. The challenge now is how to render the U-relations $\boldsymbol Y$. 
Before proceeding, we consider Example \ref{ex:population}, which is admittedly small but fairly representative to illustrate how to deal with correlations in the predictive data of deterministic hypotheses.

\begin{figure*}[t]
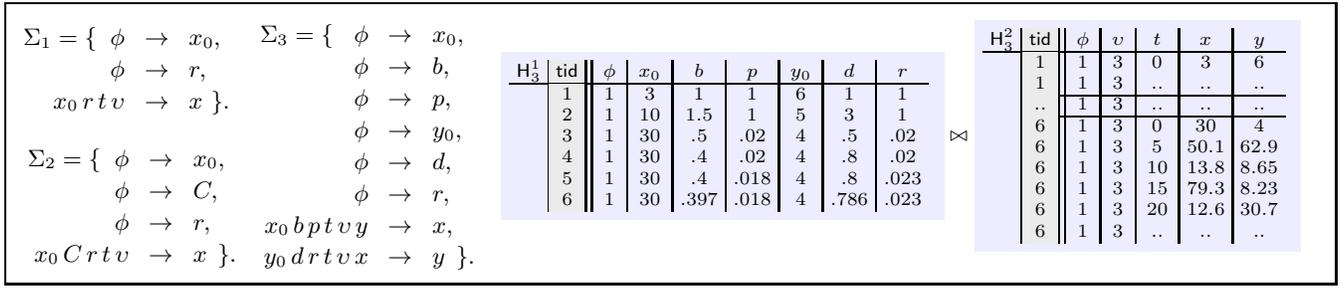

\footnotesize
\begin{framed}
\vspace{-9pt}
\begin{subfigure}{0.125\columnwidth}
\begin{eqnarray*}
\!\Sigma_1 = \{\;\; \phi &\!\to\!& x_0,\\ 
\phi &\!\to\!& r,\\ 
\!\!x_0\,r\,t\,\upsilon &\!\to\!& x \,\,\}.
\end{eqnarray*}
\begin{eqnarray*}
\Sigma_2 = \{\;\; \phi &\!\to\!& x_0,\\ 
\phi &\!\to\!& C,\\ 
\phi &\!\to\!& r,\\ 
x_0\,C\,r\,t\,\upsilon &\!\to\!& x \;\,\}.
\end{eqnarray*}
\end{subfigure}
\hspace{15pt}
\begin{subfigure}{0.2\columnwidth}
\begin{eqnarray*}
\Sigma_3 = \{\;\;\; \phi &\!\to\!& x_0,\\ 
\phi &\!\to\!& b,\\
\phi &\!\to\!& p,\\
\phi &\!\to\!& y_0,\\ 
\phi &\!\to\!& d,\\ 
\phi &\!\to\!& r,\\
x_0\,b\,p\,t\,\upsilon\,y &\!\to\!& x,\\ 
y_0\,d\,r\,t\,\upsilon\,x &\!\to\!& y \;\,\}.
\end{eqnarray*}
\end{subfigure}
\begin{subfigure}{1.38\columnwidth}
\scriptsize
\begingroup\setlength{\fboxsep}{3pt}
\colorbox{blue!7}{%
   \begin{tabular}{c|>{\columncolor[gray]{0.92}}c||c|c|c|c|c|c|c}
  $\!\!$\textsf{H}$_3^1\!\!$ & $\!\!$\textsf{tid}$\!\!$ & $\!\phi\!$ & $\!x_0\!$ & $b$ & $p$ & $\!y_0\!$ & $d$ & $r$\\
      \hline
   & $\!\!1\!\!$ & $\!1\!$ & $\!3\!$ & $\!\!1\!\!$ & $\!\!1\!\!$ & $\!\!6\!\!$ & $\!\!1\!\!$ & $\!\!1\!\!$\\
   & $\!\!2\!\!$ & $\!1\!$ & $\!10\!$ & $\!\!1.5\!\!$ & $\!\!1\!\!$ & $\!\!5\!\!$ & $\!\!3\!\!$ & $\!\!1\!\!$\\
   & $\!\!3\!\!$ & $\!1\!$ & $\!30\!$ & $\!\!.5\!\!$ & $\!\!.02\!\!$ & $\!\!4\!\!$ & $\!\!.5\!\!$ & $\!\!.02\!\!$\\
   & $\!\!4\!\!$ & $\!1\!$ & $\!30\!$ & $\!\!.4\!\!$ & $\!\!.02\!\!$ & $\!\!4\!\!$ & $\!\!.8\!\!$ & $\!\!.02\!\!$\\
   & $\!\!5\!\!$ & $\!1\!$ & $\!30\!$ & $\!\!.4\!\!$ & $\!\!.018\!\!$ & $\!\!4\!\!$ & $\!\!.8\!\!$ & $\!\!.023\!\!$\\
   & $\!\!6\!\!$ & $\!1\!$ & $\!30\!$ & $\!\!.397\!\!$ & $\!\!.018\!\!$ & $\!\!4\!\!$ & $\!\!.786\!\!$ & $\!\!.023\!\!$\\
   \end{tabular}
}\endgroup
$\;\bowtie\!\!$
\begingroup\setlength{\fboxsep}{3pt}
\colorbox{blue!7}{%
   \begin{tabular}{c|>{\columncolor[gray]{0.92}}c||c|c|c|c|c}
  $\!\!$\textsf{H}$_3^2\!\!$ & $\!\!\textsf{tid}\!\!$ & $\!\phi\!$ & $\!\upsilon\!$ & $\! t \!$ & $x$ & $y$\\
      \hline    
   & $\!\!1\!\!$ & $\!1\!$ & $\!3\!$ & $\!0\!$ & $\!\!3\!\!$ & $\!\!6\!\!$\\
   & $\!\!1\!\!$ & $\!1\!$ & $\!3\!$ & $\!..\!$ & .. & ..\\
   \cline{2-7}
   & $\!\!..\!\!$ & $\!1\!$ & $\!3\!$ & $\!..\!$ & .. & ..\\
   \cline{2-7}
   & $\!\!6\!\!$ & $\!1\!$ & $\!3\!$ & $\!0\!$ & $\!\!30\!\!$ & $4$\\
   & $\!\!6\!\!$ & $\!1\!$ & $\!3\!$ & $\!5\!$ & $\!\!50.1\!\!$ & $\!\!62.9\!\!$\\
   & $\!\!6\!\!$ & $\!1\!$ & $\!3\!$ & $\!10\!$ & $\!\!13.8\!\!$ & $\!\!8.65\!\!$\\
   & $\!\!6\!\!$ & $\!1\!$ & $\!3\!$ & $\!15\!$ & $\!\!79.3\!\!$ & $\!\!8.23\!\!$\\
   & $\!\!6\!\!$ & $\!1\!$ & $\!3\!$ & $\!20\!$ & $\!\!12.6\!\!$ & $\!\!30.7\!\!$\\
   & $\!\!6\!\!$ & $\!1\!$ & $\!3\!$ & $\!..\!$ & .. & ..\\
   \end{tabular}
}\endgroup
\end{subfigure}\vspace{5pt}\\
\vspace{-9pt}
\end{framed}
\vspace{-8pt}
\caption{Resources from Example \ref{ex:population}. (Left). Primitive fd sets extracted from the given structures $\mathcal S_k(\mathcal E_k,\, \mathcal V_k)$ for $k=1..3$. (Right). \emph{Certain} relations $\boldsymbol H_3\!=\! \{ H_3^1,\; H_3^2 \}$ of hypothesis $\upsilon\!=\!3$ loaded with trial datasets identified by special attribute \textsf{tid}.}
\label{fig:population}
\vspace{-3pt}
\end{figure*}

\begin{myex}
We explore three slightly different theoretical models in population dynamics with applications in Ecology, Epidemics, Economics, etc: (\ref{eq:malthus}) Malthus' model, (\ref{eq:logistic})$\!$ the logistic equation and (\ref{eq:lotka-volterra}) the Lotka-$\!$Volterra model. In practice, such equations are meant to be extracted from \textsf{MathML}-compliant XML files (cf. \S\ref{subsec:applicability}). For now, consider that the ordinary differential equation notation `$\dot{x}$' is read `variable $x$ is a function of time $t$ given initial condition $x_0$.$\!$' 
\begin{eqnarray}
\dot{x}=rx
\label{eq:malthus}
\end{eqnarray}
\vspace{-18pt}
\begin{eqnarray}
\dot{x}=r(C-x)x
\label{eq:logistic}
\end{eqnarray}
\vspace{-18pt}
\begin{eqnarray}
\left\{ 
  \begin{array}{lll}
\dot{x} &=& x(b - py)\\
\dot{y} &=& y(rx - d)
\end{array} \right.
\label{eq:lotka-volterra}
\end{eqnarray}
The models are completed (by the user) with additional equations to provide the values of exogenous variables (or ``input parameters''),\footnote{Given $\mathcal S(\mathcal E, \mathcal V)$, it is actually a task of the encoding algorithm (viz., \textsf{COA}$_t$'s) to infer whether its variables $x \in \mathcal V$ are exogenous or endogenous by processing its causal ordering.} e.g., $x_0\!=\!200,\, r\!=\!10$, such that we have structures $\mathcal S_k(\mathcal E_k, \mathcal V_k)$, for $k=1..3$,
\begin{itemize}
\item $\mathcal E_1 \!=\! \{\, f_1(t),\; f_2(x_0),\; f_3(r),\; f_4(x, t, x_0, r) \,\}$;
\item $\mathcal E_2 \!=\! \{ f_1(t),\, f_2(x_0),\, f_3(C),\, f_4(r),\, f_5(x, t, x_0, C, r) \}$;
\item $\mathcal E_3 \!=\! \{\, f_1(t),\; f_2(x_0),\; f_3(b),\; f_4(p),\; f_5(y_0),\; f_6(d),$\vspace{2pt}\\ 
			$f_7(r),\; f_8(x, t, x_0, b, p, y),\; f_9(y, t, y_0, d, r, x) \,\}$. 
\end{itemize}

\noindent
Fig. $\!$\ref{fig:population} shows the fd sets encoded from structures $\mathcal S_k$ above.\footnote{Domain variables like time $t$ require a special treatment by \textsf{h-encode} to suppress an fd $\phi \!\to t$. This is coped with by providing it an additional argument $\ell$ informing that the ($\ell\!\times\!\ell$)-first block of matrix $A_S$ is kept for domain variables.}
We shall also consider trial datasets for hypothesis $\upsilon\!=\!3$ (viz., the Lotka-Volterra model), which are loaded into the synthesized (certain) schemes in $\boldsymbol H_3$ as shown in Fig. $\!$\ref{fig:population}. Note that the fd's in $\Sigma_3$ are violated by relations $H_3^1, \,H_3^2$, but we admit a special attribute `trial id' \textsf{tid} into their key constraints for a trivial repair (provisionally, yet at the \textsf{ETL} stage of the pipeline) until uncertainty is introduced in a controlled way by synthesis `4U' (\textsf{U-intro} stage, cf. Fig. $\!$\ref{fig:pipeline}).  $\Box$
\label{ex:population}
\end{myex}

\noindent
Given \emph{certain} relations $\boldsymbol H$, synthesis `4U' has two parts: process the uncertainty of exogenous relations (\emph{u-factor\-}\emph{ization}) and of endogenous relations (\emph{u-propagation}).

\subsection{U-Factorization}\label{subsec:u-factors}
\vspace{-2pt}
\noindent
As we have seen in \S\ref{subsec:u-relations}, the \textsf{repair-key} operation allows one to create a discrete random variable in order to repair an argument key in a given relation. Our goal here is to devise a technique to perform such operation in a principled way for hypothesis management. It is a basic design principle to have exactly one random variable for each distinct uncertainty factor (`u-factor' for short), which requires carefully identifying the actual sources of uncertainty present in relations $\boldsymbol H$. 

The multiplicity of (competing) hypotheses is itself a standard one, viz., the \emph{theoretical} u-factor. Consider an `explanation' table like $H_0$ in Fig. $\!$\ref{fig:maybms}, which stores (as foreign keys) all hypotheses available and their target phenomena. We can take such $H_0$ as explanation table for the three hypotheses of Example \ref{ex:population}. Then a discrete random variable \textsf{x$_0$} (not to be confused with variable a $x_0 \!\in\! \mathcal V$) is defined into $Y_0[\,V_0D_0\,|\,\phi\, \upsilon\,]$ by query formula (\ref{eq:explanation}). U-relation $Y_0$ is considered standard in synthesis `4U,' as the repair of $\phi$ as a key in (standard) $H_0$.

Hypotheses, though, are (abstract) `universal state\-ments' \cite{losee2001}. In order to produce a (concrete) valuation over their endogenous attributes (predictions), one has to inquire into some particular `situated' phenomenon $\phi$ and tentatively assign a valuation over the exogenous attributes, which can be eventually tuned for a target $\phi$. The multiplicity of such (competing) empirical estimations for a hypothesis $k$ leads to Problem \ref{prob:u-learning}, viz., learning \emph{empirical} u-factors for each  $\boldsymbol H_k \subseteq \boldsymbol H$.

\begin{myprob}
Let $H_k^\ell[XZ] \in \boldsymbol H_k$ be an exogenous relation with key constraint $X \!\to Z$. Once $H_k^\ell$ is loaded with trial data, the problem of \textbf{u-factor learning} is:
\begin{enumerate}
\item to \textbf{infer} in $H_k^\ell$ ``casual'' fd's $B_i \!\leftrightarrow\! B_j \notin \Sigma_k\!$ (strong input correlations), 
where $B_i,\, B_j \in Z$; 
\item to form maximal \textbf{groups} $G_1,\,...,\,G_n \subseteq Z$ of attributes such that for all $B_i,\, B_j \in G_a$, the casual fd's $B_i \!\leftrightarrow\! B_j$ hold in $H_k^\ell$;
\item to pick, for each group $G_a$, any $A \in G_a$ as a \textbf{pivot} representative and insert $A \!\to B$ into an fd set $\Gamma_k$ for all $B \in (G_a \setminus A)$.
\end{enumerate}
\label{prob:u-learning}
\vspace{-6pt}
\end{myprob}

\noindent
Problem \ref{prob:u-learning} is dominated by the (problem of) discovery of fd's in a relation, which is not really a new problem (e.g., see \cite{huhtala1999}). We then keep focus on the synthesis `4U' as a whole and omit our detailed \textsf{u-factor-learning} algorithm in particular. Its output, fd set $\Gamma_k$, is then filled in (completed) with the $\upsilon$-fd's from $\Sigma_k$.

For illustration consider hypothesis $\upsilon\!=\!3$ and its trial input data recorded in $H^1_3$ in Fig. $\!$\ref{fig:population}. We show its corresponding fd set $\Gamma_3$ in Fig. $\!$\ref{fig:y-fdschema} (left). Recall that, as a result of synthesis `4C,' relation $H^1_3$ is in BCNF w.r.t. $\!\Sigma_3^\looparrowright\!$. Since its attributes have been inferred exogenous in the given hypothesis (cf. $\!$Proposition $\!$\ref{prop:exo-endo}), they are then officially unrelated. In fact, by ``casual'' fd's we mean correlations that, for a set of experimental trials, may occasionally show up in the trial input data --- e.g., $x_0 \leftrightarrow y_0$ hold in $H^1_3$, but not because $x_0$ and $y_0$ are related in principle (theory).

\begin{figure}[t]
\begin{framed}
\vspace{-11pt}
\begin{subfigure}{0.45\columnwidth}
\begin{eqnarray*}
\Gamma_3 = \{\;\;\; x_0 &\to& y_0,\\ 
b &\to& d,\\ 
p &\to& r,\\ 
x_0\,b\,p\,t\,\upsilon\,y &\!\to\!& x,\\ 
y_0\,d\,r\,t\,\upsilon\,x &\!\to\!& y \;\;\;\}.
\end{eqnarray*}
\end{subfigure}
\hspace{12pt}
\begin{subfigure}{0.45\columnwidth}
\begin{eqnarray*}
\Gamma_3^\looparrowright  = \{\;\;\; x_0 &\to& y_0,\\ 
b &\to& d,\\ 
p &\to& r,\\ 
x_0\,b\,p\,t\,\upsilon\,y &\!\to\!& x,\\ 
x_0\,b\,p\,t\,\upsilon\,x &\!\to\!& y \;\;\;\}.
\end{eqnarray*}
\end{subfigure}
\end{framed}
\vspace{-11pt}
\caption{Fd set $\Gamma_3$ (compare with $\Sigma_3$) and its folding $\Gamma_3^\looparrowright\!$.}
\label{fig:y-fdschema}
\vspace{-6pt}
\end{figure}

Once fd set $\Gamma_k$ is output by u-factor learning, its folding $\Gamma_k^\looparrowright\!$ shall be given with $\boldsymbol H_k$ as input to accomplish u-factorization for hypothesis $k$ algorithmically. We shall employ a notion of u-factor decomposition formulated in Def. $\!$\ref{def:u-factor} into query formula (\ref{eq:u-factor}).

\begin{mydef}
\label{def:u-factor}
$\!$Let $R[A_p W] \!\in\! \boldsymbol R$ be an exogenous scheme\vspace{-2pt} with \mbox{$\boldsymbol R\!\triangleq$ \textsf{synthesize}$(\Phi)$} designed over subset $\Phi\!$ of $\phi$-fd's in\vspace{-1pt} $\Gamma_k^\looparrowright\!$; and let $H_k^\ell[X_\ell L] \in \boldsymbol H_k$ be an exogenous relation with (violated) key constraint $\langle X_\ell , L\rangle \in \Sigma_k^\looparrowright\!$ with $A_pW \!\subseteq L$. $\!$Then the exogenous U-relation $Y_k^p[V_pD_p\,|\,X_\ell A_p]$ for sketched scheme $R[A_p W]$ is defined by query formula (\ref{eq:u-factor}) in p-WSA's extension of relational algebra, 
\begin{eqnarray}
\label{eq:u-factor}
\! Y_k^p := \pi_{X_\ell A_p}( \textsf{repair-key}_{X_\ell @\textsf{\emph{count}}} (\,\gamma_{X_\ell A_p,\,\textsf{\emph{count}}(*)}(H_k^\ell)\,)\,)
\end{eqnarray}
where $\gamma$ is relational algebra's grouping operator. Let\vspace{1pt} $G_a = A_p W$. We say that $Y_k^p$ is a \textbf{u-factor projection} of $H_k^\ell[X_\ell L]$ if $A_p \leftrightarrow B$ hold in $H_k^\ell$ for all $B \in G_a$ and for no $C \in (L \setminus G_a)$ we have $C \to B$ or\vspace{1.5pt} $B \to C$. 
\end{mydef}

\begin{myprop}
\label{prop:u-factor}
$\!$Let $H_k^\ell[X_\ell L] \!\in \boldsymbol H_k$ be an exogenous relation with (violated) key constraint $\langle X_\ell, \!L\rangle \!\in\! \Sigma_k^\looparrowright\!$. $\!$Then,$\!$
\begin{itemize}
\item[(a)] for any pair $Y_k^i[\,V_iD_i\,|\,X_\ell A_i],\; Y_k^j[\,V_jD_j\,|\,X_\ell A_j]$ of u-factor projections of $H_k^\ell$, they are independent.
\item[(b)] the join $\bowtie_{i=1}^m \! Y_k^i[\,V_iD_i\,|\,X_\ell A_i\,]$ of all u-factor projections  of $H_k^\ell$ is lossless w.r.t. $\!\pi_{X_\ell A_1\,A_2\,...\,A_m}(\Sigma_k^\looparrowright\!)$.
\end{itemize}
\end{myprop}
\begin{proof}
See Appendix, \S\ref{a:u-factor}.
\end{proof}

\vspace{3pt}
Proposition \ref{prop:u-factor} is significant as it ensures all the empirical uncertainty implicit in an exogenous relation can be decomposed into u-factor projections that are (a) in fact independent, to do justice to the term `factors,' and (b) and can be fully recovered by a lossless join. 
The u-factorization procedure given fd set $\Gamma_k^\looparrowright\!$ and relations $\boldsymbol H_k$ is described by (Alg. \ref{alg:synthesize4u}) \textsf{synthesize4u} (Part I).

\vspace{2pt}
\subsection{U-Propagation}\label{subsec:u-propagation}
\vspace{-2pt}
\noindent
For hypothesis $\upsilon\!=\!k$, take some endogenous attribute $B$, and then by Def. $\!$\ref{def:exo-endo} and the parsimony assumption we must have exactly one $\upsilon$-fd $S \!\to B$ in $\Gamma_k^\looparrowright\!\!$. Then note that the u-factors with `incidence' on $B$ are in $S$. In comparison with synthesis `4C,' we have just unfolded $B$'s ``causal chain'' out of its compact form in $\Sigma_k^\looparrowright\!$ and re-folded it\vspace{-1pt} fine-grained (over u-factor pivots) into $\Gamma_k^\looparrowright\!$. 

Note that the $\upsilon$-fd's in $\Gamma_k^\looparrowright\!\!$, of form $S \!\to B$, are meant for u-propagation. Each pivot attribute $A_j \in S$ shall be used as a surrogate to its associated random variable $\textsf{x}_j$ from exogenous U-relation $Y_k^j[V_jD_j \,|\,X_\ell A_j]$ to propagate uncertainty properly into endogenous U-relations\vspace{-1pt} $Y_k^r[\overline{V_iD_i} \,|\,ZT]$, for $B \in T$ and each $A_j \in S \setminus Z$. This intuition is abstracted into a general p-WSA query formula (\ref{eq:u-propagation}) as given in Def. $\!$\ref{def:u-propagation}, and employed in (Alg. \ref{alg:synthesize4u}) \textsf{synthesize4u} to accomplish u-propagation (Part II).

\begin{algorithm}[t]\footnotesize
\caption{Synthesis `4U' applied over folding fd set.}
\label{alg:synthesize4u}
\begin{algorithmic}[1]
\Procedure{synthesize4u}{$\Gamma_k^\looparrowright\!\!: \text{fd set}$, $\boldsymbol H_k\!: \text{DB}$}
\Require $\Gamma_k^\looparrowright\!$ is the folding of parsimonious fd set $\Gamma_k$\vspace{1pt}
\Ensure U-relational DB $\boldsymbol Y_k$ returned is $\boldsymbol H_k$ after U-intro
\State $\Phi \gets \varnothing,\; \Upsilon \gets \varnothing$
\ForAll{$\langle X, B\rangle \in \Gamma_k^\looparrowright\!$}
\If{$\upsilon \notin X$}
\State $\Phi \gets \Phi  \cup \langle X, B\rangle$ \Comment{$\phi$-fd over exogenous $B$}
\Else{\State $\Upsilon \gets \Upsilon  \cup \langle X, B\rangle$} \Comment{$\upsilon$-fd over endogenous $B$}
\EndIf
\EndFor
\vspace{-2pt}
\Statex{}
\hrulefill
\vspace{2pt}
\Statex{Part I: \textbf{U-factorization}}
\vspace{2pt}
\State $\mathsf M \gets \varnothing$ \Comment{store u-factor projection mappings}
\State $\boldsymbol R \gets \;\textsf{synthesize}(\Phi)\!$ \Comment{design BCNF exog. schemes}\vspace{3pt}
\ForAll{$R[A_p W] \in \boldsymbol R$}
\State find exog. $H_k^\ell[X_\ell L] \in \boldsymbol H_k\;\text{such that}\; A_p \in L$
\State $Y_k^i \gets \pi_{X_\ell A_p}( \textsf{repair-key}_{X_\ell @\textsf{\emph{count}}} (\,\gamma_{X_\ell A_p,\,\textsf{\emph{count}}(*)}(H_k^\ell)\,)\,)$
\State $\boldsymbol Y_k \gets \boldsymbol Y_k \cup Y_k^i$
\If{$H_k^\ell \in \mathsf M$} \Comment{save mapping for further ref.}
\State $\mathsf M(H_k^\ell) \gets \mathsf M(H_k^\ell) \cup \{Y_k^i\}$
\Else
\State $\mathsf M \gets \mathsf M \cup \langle H_k^\ell, \{Y_k^i\} \rangle$
\EndIf
\EndFor
\vspace{-2pt}
\Statex{}
\hrulefill
\vspace{2pt}
\Statex{Part II: \textbf{U-propagation}}
\vspace{2pt}
\State $\boldsymbol R \gets \;\textsf{synthesize}(\Upsilon)\!$ \Comment{design BCNF endog. schemes}\vspace{3pt}
\ForAll{$R[ST] \in \boldsymbol R$}
\State $X \gets \varnothing$, $J \gets \varnothing$ \Comment{prepare for join sub-query}
\ForAll{$H_k^\ell[X_\ell\, L] \in \mathsf M$}\vspace{2pt}
\If{$L \cap S \neq \varnothing$} 
\State $J \gets J \bowtie H_k^\ell$
\State $X \gets X \cup X_\ell$
\ForAll{$Y_k^i[X_\ell A_p] \in \mathsf M(H_k^\ell)$}
\If{$A_p \in S$} \Comment{$Y_k^i$ is a u-factor for $R$}
\State $J \gets J \bowtie Y_k^i$
\EndIf
\EndFor
\EndIf
\EndFor
\State find $H_k^q[Z_q V] \in \boldsymbol H_k$ such that $V \supseteq T$\vspace{1pt}
\State $Y_k^r \gets \pi_{Z_q T}(\, \sigma_{\upsilon=k}(Y_0) \bowtie \pi_{\textsf{tid},X}(J) \bowtie \pi_{\textsf{tid},Z_q T}(H_k^q)\,) \!\!\!$\vspace{2pt}
\State $\boldsymbol Y_k \gets \boldsymbol Y_k \cup Y_k^r$
\EndFor
\State \Return $\boldsymbol Y_k$
\EndProcedure
\end{algorithmic}
\end{algorithm}

\begin{mydef}
\label{def:u-propagation}
$\!$Let $R[ST] \!\in\! \boldsymbol R$ be an endogenous scheme,\vspace{-2pt} $\!$designed \mbox{$\boldsymbol R \triangleq$ \textsf{synthesize}$(\Upsilon)$} over subset $\Upsilon$ of $\upsilon$-fd's in $\Gamma_k^\looparrowright\!$; and let $H_k^q[Z_q V] \in \boldsymbol H_k$ be an endogenous relation\vspace{-.75pt} with (violated) key constraint $\langle Z_q , V\rangle \in \Sigma_k^\looparrowright\!$ such that $V \supseteq T$.
 $\!$Then the endogenous U-relation $Y_k^{q}[\overline{V_iD_i} \,|\,Z_q T]$ for sketched scheme $R[ST]$ is\vspace{1pt} defined by formula (\ref{eq:u-propagation}),
\begin{eqnarray}
\label{eq:u-propagation}
Y_k^{r} :=\, \pi_{Z_q T}(\, \sigma_{\upsilon=k}(Y_0) \bowtie \pi_{\textsf{tid},X}(J) \bowtie \pi_{\textsf{tid},Z_q T}(H_k^q)\,)
\end{eqnarray}
\noindent
where $J$ is a join sub-query defined over mapping $\mathsf M$\vspace{1pt} from exogenous relations to sets of exogenous U-relations:
\begin{itemize}
\item[(a)] $\!$for $H_k^\ell \in \boldsymbol H_k$, we have $H_k^\ell[X_\ell L]  \in\mathsf M$ iff $L \cap S \neq \varnothing$;\vspace{-3pt}
\item[(b)] $\!$we take $X = \bigcup_{H_k^\ell \in \mathsf M}X_\ell$;\vspace{-5pt}
\item[(c)] $\!$we have $Y_k^j[V_jD_j | X_\ell A_j] \in \mathsf M(H_k^\ell)$ iff $Y_k^j$ is a u-factor projection of $H_k^\ell$ with $A_j \in (L \cap S)$.
\end{itemize}
We say that $Y_k^{r}$ is a \textbf{predictive projection} of $H_k^q$.
\end{mydef}

\begin{mythm}\label{thm:u-propagation}
Let $H_k^q[Z_q V] \!\in\! \boldsymbol H_k$ be an endogenous re\-la\-tion with (violated) key constraint $\langle Z_q, V\rangle\! \in \Sigma_k^\looparrowright\!$, and $Y_k^r[\overline{V_iD_i}\,|\,Z_qT]$ be a predictive projection of $H_k^q$ w.r.t. $\Gamma_k^\looparrowright\!$ defined by formula (\ref{eq:u-propagation}) with $V \!\supseteq T$. $\!$We claim that $Y_k^r$ correctly captures all the uncertainty present in $H_k^q$ w.r.t. $\Gamma_k^\looparrowright\!$.$\!$
\end{mythm}
\begin{proof}
See Appendix, \S\ref{a:u-propagation}.
\end{proof}

Fig. \ref{fig:u-relations} shows the rendered U-relations for hypothesis $\upsilon\!=\!3$ whose relations are shown in Fig. \ref{fig:population}. Note that $\textsf{tid}\!=\!6$ in $H_3^2$ (Fig. $\!$\ref{fig:population}) corresponds now to $\theta = \{\,\textsf{x}_0 \!\mapsto\! 3,\, \textsf{x}_1 \!\mapsto\! 3,$ $\textsf{x}_2 \!\mapsto\! 5,\, \textsf{x}_3 \!\mapsto\! 3 \,\}$, where $\theta$ defines a particular world in $\boldsymbol W$ whose probability is $\textsf{Pr}(\theta) \!\approx\! .012$. This value is derived from the marginal probabilities stored in world table $W$ (e.g., see Fig. $\!$\ref{fig:maybms}) as a result of the application of formulas (\ref{eq:explanation}) and (\ref{eq:u-factor}).

\begin{myremark}
Observe that, although (Alg. \ref{alg:synthesize4u}) \textsf{synthesize4u} operates locally for each hypothesis $k$, the effects of synthesis `4U' (\textsf{U-intro}) in the pipeline are global on account of the (global) `explanation' relation $H_0$ (then U-relation $Y_0$); e.g., see Fig. \ref{fig:u-relations}. In fact, the probability of each tuple, say, in endogenous U-relation $Y_k^q$ with $\phi=p$ for hypothesis $\upsilon=k$, is distributed among all the hypotheses $\ell\neq k$ that are keyed in $Y_0$ under $\phi=p$. $\Box$
\label{rmk:u-intro}
\end{myremark}

\begin{figure}[t]
\centering
\scriptsize
\begingroup\setlength{\fboxsep}{2pt}
\colorbox{yellow!10}{%
   \begin{tabular}{c|>{\columncolor[gray]{0.92}}c||c|c}
  $\!\!$\textsf{Y}$_3^1\!\!$ & $\!\!\!V \!\mapsto\! D\!\!\!$ & $\!\!\!\phi\!\!\!$ & $\!\!\!A_1\!\!\!$\\
      \hline    
   & $\!\!\textsf{x}_1 \!\mapsto\! 1\!\!$ & $\!\!1\!\!$ & $\!\!3\!\!$\\
   & $\!\!\textsf{x}_1 \!\mapsto\! 2\!\!$ & $\!\!1\!\!$ & $\!\!10\!\!$\\
   & $\!\!\textsf{x}_1 \!\mapsto\! 3\!\!$ & $\!\!1\!\!$ & $\!\!30\!\!$\\
   \end{tabular}
}\endgroup
$\!\!\!$
\begingroup\setlength{\fboxsep}{2pt}
\colorbox{yellow!10}{%
   \begin{tabular}{c|>{\columncolor[gray]{0.92}}c||c|c}
  $\!\!$\textsf{Y}$_3^2\!\!$ & $\!\!\!V \!\mapsto\! D\!\!\!$ & $\!\!\phi\!\!$ & $\!\!\!A_2\!\!\!$\\
      \hline    
   & $\!\!\textsf{x}_2 \!\mapsto\! 1\!\!$ & $\!1\!$ & $\!\!1\!\!$\\
   & $\!\!\textsf{x}_2 \!\mapsto\! 2\!\!$ & $\!1\!$ & $\!\!1.5\!\!$\\
   & $\!\!\textsf{x}_2 \!\mapsto\! 3\!\!$ & $\!1\!$ & $\!\!.5\!\!$\\
   & $\!\!\textsf{x}_2 \!\mapsto\! 4\!\!$ & $\!1\!$ & $\!\!.4\!\!$\\
   & $\!\!\textsf{x}_2 \!\mapsto\! 5\!\!$ & $\!1\!$ & $\!\!\!.397\!\!\!$\\
   \end{tabular}
}\endgroup
\begingroup\setlength{\fboxsep}{2pt}
\colorbox{yellow!10}{%
   \begin{tabular}{c|>{\columncolor[gray]{0.92}}c||c|c}
  $\!\!$\textsf{Y}$_3^3\!\!$ & $\!\!\!V \!\mapsto\! D\!\!\!$ & $\!\!\phi\!\!$ & $\!\!\!A_3\!\!\!$\\
      \hline    
   & $\!\!\textsf{x}_3 \!\mapsto\! 1\!\!$ & $\!1\!$ & $\!\!1\!\!$\\
   & $\!\!\textsf{x}_3 \!\mapsto\! 2\!\!$ & $\!1\!$ & $\!\!.02\!\!$\\
   & $\!\!\textsf{x}_3 \!\mapsto\! 3\!\!$ & $\!1\!$ & $\!\!\!.018\!\!\!$\\
   \end{tabular}
}\endgroup
\vspace{1pt}\\
\begingroup\setlength{\fboxsep}{3pt}
\colorbox{yellow!10}{%
   \begin{tabular}{c|>{\columncolor[gray]{0.92}}c|>{\columncolor[gray]{0.92}}c|>{\columncolor[gray]{0.92}}c|>{\columncolor[gray]{0.92}}c||c|c|c|c|c}
  $\!\!$\textsf{Y}$_3^4\!\!$ & $\!\!\!V_0 \!\mapsto\! D_0\!\!\!$ & $\!\!\!V_1 \!\mapsto\! D_1\!\!\!$ & $\!\!\!V_2 \!\mapsto\! D_2\!\!\!$ & $\!\!\!V_3 \!\mapsto\! D_3\!\!\!$ & $\!\phi\!$ & $\!\upsilon\!$ & $\! t \!$ & $x$ & $y$\\
      \hline    
   & $\!\!\textsf{x}_0 \mapsto 3\!\!$ & $\!\!\textsf{x}_1 \mapsto 1\!\!$ & $\!\!\textsf{x}_2 \mapsto 1\!\!$ & $\!\!\textsf{x}_3 \mapsto 1\!\!$ & $\!1\!$ & $\!3\!$ & $\!0\!$ & $\!\!3\!\!$ & $\!\!6\!\!$\\
   & $\!\!\textsf{x}_0 \mapsto 3\!\!$ & $\!\!\textsf{x}_1 \mapsto 1\!\!$ & $\!\!\textsf{x}_2 \mapsto 1\!\!$ & $\!\!\textsf{x}_3 \mapsto 1\!\!$ & $\!1\!$ & $\!3\!$ & $\!...\!$ & ... & ...\\
   \cline{2-10}
   & $\!\!...\!\!$ & $\!\!...\!\!$ & $\!\!...\!\!$ & $\!\!...\!\!$ & $\!1\!$ & $\!3\!$ & $\!...\!$ & ... & ...\\
   \cline{2-10}
   & $\!\!\textsf{x}_0 \mapsto 3\!\!$ & $\!\!\textsf{x}_1 \mapsto 3\!\!$ & $\!\!\textsf{x}_2 \mapsto 5\!\!$ & $\!\!\textsf{x}_3 \mapsto 3\!\!$ & $\!1\!$ & $\!3\!$ & $\!0\!$ & $\!\!30\!\!$ & $4$\\
   & $\!\!\textsf{x}_0 \mapsto 3\!\!$ & $\!\!\textsf{x}_1 \mapsto 3\!\!$ & $\!\!\textsf{x}_2 \mapsto 5\!\!$ & $\!\!\textsf{x}_3 \mapsto 3\!\!$ & $\!1\!$ & $\!3\!$ & $\!5\!$ & $\!\!50.1\!\!$ & $\!\!62.9\!\!$\\
   & $\!\!\textsf{x}_0 \mapsto 3\!\!$ & $\!\!\textsf{x}_1 \mapsto 3\!\!$ & $\!\!\textsf{x}_2 \mapsto 5\!\!$ & $\!\!\textsf{x}_3 \mapsto 3\!\!$ & $\!1\!$ & $\!3\!$ & $\!10\!$ & $\!\!13.8\!\!$ & $\!\!8.65\!\!$\\
   & $\!\!\textsf{x}_0 \mapsto 3\!\!$ & $\!\!\textsf{x}_1 \mapsto 3\!\!$ & $\!\!\textsf{x}_2 \mapsto 5\!\!$ & $\!\!\textsf{x}_3 \mapsto 3\!\!$ & $\!1\!$ & $\!3\!$ & $\!15\!$ & $\!\!79.3\!\!$ & $\!\!8.23\!\!$\\
   & $\!\!\textsf{x}_0 \mapsto 3\!\!$ & $\!\!\textsf{x}_1 \mapsto 3\!\!$ & $\!\!\textsf{x}_2 \mapsto 5\!\!$ & $\!\!\textsf{x}_3 \mapsto 3\!\!$ & $\!1\!$ & $\!3\!$ & $\!20\!$ & $\!\!12.6\!\!$ & $\!\!30.7\!\!$\\
   & $\!\!\textsf{x}_0 \mapsto 3\!\!$ & $\!\!\textsf{x}_1 \mapsto 3\!\!$ & $\!\!\textsf{x}_2 \mapsto 5\!\!$ & $\!\!\textsf{x}_3 \mapsto 3\!\!$ & $\!1\!$ & $\!3\!$ & $\!..\!$ & .. & ..\\
   \end{tabular}
}\endgroup
\vspace{-4pt}
\caption{U-relations rendered by synthesis `4U' for hypothesis $\upsilon=3$ from Example \ref{ex:population}.}
\label{fig:u-relations}
\vspace{-8pt}
\end{figure}

In sum, the synthesis `4U' technique completes the pipeline (Fig. $\!$\ref{fig:pipeline}), except for the problem of conditioning (cf. \cite{goncalves2014}) which is not covered in this paper.

\vspace{-2pt}
\section{Discussion}\label{sec:discussion}
\vspace{-3pt}
\noindent
In this section we discuss related work and the applicability of our framework.
\vspace{-2pt}
\subsection{Related Work}\label{subsec:related-work}
\vspace{-2pt}
\textbf{Models and data}. 
Haas et al. \cite{haas2011} propose a long-term \emph{models-and-data} research program to address data management for \emph{deep} predictive analytics. They discusss strategies to extend query engines for model execution within a \mbox{(p-)DB}. Along these lines, query optimization is understood as a more general problem with connections to algebraic solvers.
Our framework in turn essentially comprises an abstraction of \emph{hypotheses as data} \cite{goncalves2014}. It can be understood in comparison as putting models strictly into a (flattened) data perspective. For this reason it is directly applicable by building upon recent work on p-DBs \cite{suciu2011}.

\textbf{Design by synthesis}. 
Classical design by synthesis \cite{bernstein1976} was once criticized due to its too strong `uniqueness' of fd's assumption \cite[p. $\!$443]{fagin1977}, as it reduces the problem of design to symbolic reasoning on fd's arguably neglecting semantic issues. For hypothesis management, however, design by synthesis is clearly feasible (see also \S\ref{subsec:applicability}), as it translates seamlessly to data dependencies the reduction made by the user herself into a tentative formal model for the studied phenomenon. In fact, synthesis methods may be as fruitful for p-DB design as they are for GM design (cf. \cite{druzdzel1993,darwiche2010}).

\textbf{Reasoning over fd's}. The concept of fd folding and design of (Alg. $\!$\ref{alg:afolding}) \textsf{AFolding} as a variant of \textsf{XClosure}, is (arguably) a non-obvious approach to the problem of processing the causal ordering of a hypothesis via acyclic reflexive pseudo-transitive reasoning. To the best of our knowledge, such a specific form of fd reasoning over fd's was an yet unexplored problem in the database research literature.

\textbf{Causality in AI}. The work of Dash and Druzdzel (cf. $\!$\cite{druzdzel2008}) is peculiar in that it explores SEMs to model deterministic (quantitative) hypotheses into GMs (like in Simon's early work), while a majority of work in their field is devoted to model statistical (qualitative) hypotheses \cite{pearl2000,darwiche2010}. In comparison, we are applying SEMs to model deterministic hypotheses into U-relational DBs.

\textbf{Causality in DBs}. 
Our encoding of equations into fd's (constraints/correlations) captures the causal chain from exogenous (input) to endogenous (output) tuples. Fd's are rich, stronger information that can be exploited in reasoning about causality in a DB for the sake of explanation and sensitivity analysis. 
$\!$In comparison with \mbox{Meliou et. al \cite{meliou2010}} and Kanagal et. al \cite{kanagal2011}, we are processing causality at schema level. To our knowledge this is the first work to address causal reasoning in the presence of constraints (viz., fd's), called for as worth of future work by Meliou et al.$\!$ \cite[p. $\!$3]{meliou2010}.

\textbf{Hypothesis encoding}. 
Finally, our framework is comparable with Bioinformatics' initiatives that address hypothesis encoding into the RDF data model \cite{soldatova2011}.  
We point out the Robot Scientist, HyBrow and SWAN (cf. \cite{soldatova2011}), all of which consist in some ad-hoc RDF encoding of sequence and genome analysis hypotheses under varying levels of structure (viz., from `gene G has function A' statements to free text). Our framework in turn consists in the U-relational encoding of hypotheses from mathematical equations, which is (to our knowledge) the first work on hypothesis relational encoding.

\subsection{Applicability}\label{subsec:applicability}
\vspace{-2pt}
\noindent
Essentially, the formal framework proposed in this paper is meant for (large-scale) hypothesis management and predictive analytics directly in the data under support of the querying capabilities of a p-DB. 

\textbf{Realistic assumptions}. The core assumption of our framework is that the hypotheses are given in a formal specification which is encodable into a SEM that is complete (satisfies Def. $\!$\ref{def:complete}). Also, as a semantic assumption which is standard in scientific modeling, we consider a one-to-one correspondence between real-world entities and variable/at\-tribute symbols within a structure, and that all of them must appear in some of its equations/fd's. For most science use cases involving deterministic models (if not all), such assumptions are quite reasonable. It can be a topic of future work to explore business use cases as well. 

\textbf{Hypothesis learning}. $\!$The (user) method for hypothesis formation is irrelevant to our framework, as long as the resulting hypothesis is encodable into a SEM. So, a promising use case is to incorporate machine learning methods into our framework to scale up the formation/extraction of hypotheses and evaluate them under the querying capabilities of a p-DB. Con\-sider, e.g., learning the equations, say, from \textsf{Eureqa} \cite{schmidt2009}.\footnote{\url{http://creativemachines.cornell.edu/Eureqa}.}

\textbf{Qualitative hypotheses}. Although the method is primarely motivated by computational science (usually involving differential equations), in fact it is applicable to qualitative deterministic models as well. Boolean Networks, e.g., consist in sets of functions $f(x_1, x_2, \!.., x_n)$, where $f$ is a Boolean expression. $\!$Several kinds of dynamical system can be modeled in this formalism. $\!$Appli\-cations have grown out of gene regulatory network to social network and stock market predictive analytics. Even if richer semantics is considered (e.g., fuzzy logics), our encoding method is applicable likewise, as long as the equations are still deterministic.

$\!$\textbf{Model repositories}. $\!$Recent initiatives have been fostering large-scale model integration, sharing and reproducibility in the computational sciences (e.g., \cite{hines2004,chelliah2013,hunter2003}). They are growing reasonably fast on the web, (i) promoting some \textsf{MathML}-based schema as a standard for model specification, but (ii) with limited integrity and lack of support for rating/ranking competing models. For those two reasons, they provide a strong use case for our method of hypothesis management. The Physiome project \cite{hunter2003}
, e.g., is planned to integrate very large deterministic models of human physiology. A fairly simple model of the human cardiovascular system has about 630+ variables (or equations, as $\mathcal |\mathcal E| \!=\! \mathcal |\mathcal V|$).

\textbf{Initial experiments}.
We have run an applicability study comprising the whole pipeline of Fig. $\!$\ref{fig:pipeline} in a realistic use case scenario extracted from the Physiome project \cite{goncalves2014b}.
Our initial experiments have shown that the pipeline can in fact be processed efficiently. All hypotheses extracted and analyzed happened to satisfy the lossless join property (cf. $\!$Appendix, \ref{a:conjecture}).

\vspace{-6pt}
\section{Conclusions}\label{sec:conclusions}
\vspace{-4pt}
\noindent
In this paper we have presented specific technical developments over the $\Upsilon$-DB vision \cite{goncalves2014}. In short, we have shown how to encode deterministic hypotheses as uncertain data, viz., as constraints and correlations into U-relational DBs. 

Although the pipeline (Fig. \ref{fig:pipeline}) is motivated by a very concrete class of applications, it raised some non-trivial theoretical issues and required a new design-theoretic framework in view of a principled DB research solution --- not only to the specific technical problems P1-P3 (cf. $\!$\S\ref{subsec:problem}) but to the pipeline as a whole in view of enabling hypothesis management and for predictive analytics. 

This work can be understood as revisiting and making effective use of classical design theory in a modern context to address new problems.

\vspace{-4pt}
\section{Acknowledgments}
\begin{footnotesize}
\vspace{-3pt}
\noindent
This research has been supported by the Brazilian funding agencies CNPq (grants $\!$n$^o\!$ 141838/2011-6, 309494/2012-5) and FAPERJ (grants INCT-MACC E-26/170.030/2008, `Nota $\!$10' $\!$E-26/100.286/2013). $\!$We thank IBM for a Ph.D. Fellowship 2013-2014.
\end{footnotesize}
\vspace{-6pt}

\bibliographystyle{abbrv}
\bibliography{pods15}

\newpage
\appendix

\section{Auxiliary Algorithms}\label{algs}
\noindent
We list some reference algorithms from the literature.

\begin{algorithm}[H]\footnotesize
\caption{Attribute closure $X^+$ (cf. \cite[p. $\!$388]{ullman1988}).}
\label{alg:xclosure}
\begin{algorithmic}[1]
\Procedure{XClosure}{$\Sigma\!: \text{fd set},\; X\!: \text{attribute set}$}
\Require $\Sigma$ is an fd set, $X$ is a non-empty attribute set
\Ensure $X^+$ is the attribute closure of $X$ w.r.t. $\Sigma$
\State size $\gets 0$
\State $X^+ \gets X$
\State $\Gamma \gets \Sigma$
\While{size $< |X^+|$ }
\State size $\gets |X^+|$
\ForAll{$\langle Y,\, Z \rangle \in \Gamma$}
\If{$Y \subseteq X^+$}
\State $X^+ \gets\, X^+ \cup Z$
\State $\Gamma \gets \Gamma \setminus \langle Y, Z\rangle$
\EndIf
\EndFor
\EndWhile
\State \Return $X^+$
\EndProcedure
\end{algorithmic}
\end{algorithm}

\vspace{-8pt}
\begin{algorithm}[H]\footnotesize
\caption{Left-reduced cover for a given fd set \cite{maier1983}.}
\label{alg:left-reduce}
\begin{algorithmic}[1]
\Procedure{left-reduce}{$\Sigma\!: \text{fd set}$}
\Require $\Sigma$ is an fd set
\Ensure $\Gamma$ is a left-reduced cover for $\Sigma$
\State $\Gamma \gets \Sigma$
\ForAll{$\langle X,\, Y \rangle \in \Gamma$}
\ForAll{$A \in X$}
\If{\textsf{member}($\Sigma,\; \langle X \!\setminus\! A, Y\rangle$)}
\State $\Sigma \gets \Sigma \setminus\! \langle X, Y\rangle \cup \langle X \!\setminus\! A, Y\rangle$
\EndIf
\EndFor
\EndFor
\State \Return $\Sigma$
\EndProcedure
\end{algorithmic}
\end{algorithm}

\vspace{-8pt}
\begin{algorithm}[H]\footnotesize
\caption{Membership in the closure of an fd set \cite{ullman1988}.}
\label{alg:member}
\begin{algorithmic}[1]
\Procedure{member}{$\Sigma\!: \text{fd set},\; \langle X, Y\rangle\!: \text{fd}$}
\Require $\Sigma$ is an fd set, $\langle X, Y\rangle$ is an fd
\Ensure A decision is returned
\If{$Y \subseteq \textsf{XClosure}(\Sigma,\; X$)}
\State \Return yes
\Else
\State \Return no
\EndIf
\EndProcedure
\end{algorithmic}
\end{algorithm}

\section{Detailed Proofs}

\subsection{Proof of Lemma \ref{lemma:singleton-rhs}}\label{a:singleton-rhs}
\noindent
\emph{``Let $\,\Sigma$ be a (Def. $\!$\ref{def:minimal}-a) singleton-rhs fd set on attributes $U$. Then $\langle X, A\rangle \in \Sigma^+$ with $XA \subseteq U$ only if $A \subseteq X$ or there is non-trivial $\langle Y, A\rangle \in \Sigma$ for some $Y \subset U$.''}

\begin{proof}
By Lemma \ref{lemma:ullman} (below), we know that $X \!\to A \in \Sigma^+$ iff $A \!\subseteq X^+$. We need to prove that if $A \!\nsubseteq X$ and there is no $Y \!\to A$ in singleton-rhs $\Sigma$, then $A \!\nsubseteq X^+$. But this is equivalent to show that (Alg. \ref{alg:xclosure}) \textsf{XClosure} gives only correct answers for $X^+$ w.r.t. $\Sigma$, which is known (cf. theorem from Ullman \cite[p. 389]{ullman1988}). Note that \textsf{XClosure}($\Sigma$, $X$) inserts $A$ in $X^+$ only if $A \subseteq X$ or there is some fd  $\langle Y, A\rangle \in \Sigma$.
\end{proof}

\begin{mylemma}
Let $\Sigma$ be an fd set. An fd $X \!\to Y$ is in $\Sigma^+$ iff $Y \subseteq X^+$, where $X^+$ is the attribute closure of $X$ w.r.t. $\Sigma$.
\label{lemma:ullman}
\end{mylemma}
\begin{proof}
This is from Ullman \cite[p. $\!$386]{ullman1988}. Let $Y\!=\!A_1\,...\,A_n$ and suppose $Y \subseteq X^+$. Then for each $A_i$, we have $A_i \in X^+$ and, by Def. \ref{def:xclosure}, we must have $\langle X, A_i\rangle \in \Sigma^+$. Then it follows by (R4) union that $X \to Y$ is in $\Sigma^+$ as well. Conversely, suppose $\langle X, Y\rangle \in \Sigma^+$. Then, by (R3) decomposition we have $\langle X, A_i\rangle \in \Sigma^+$ for each $A_i \in Y$.
\end{proof}

\subsection{Proof of Theorem \ref{thm:non-redundant}.}\label{a:non-redundant}
\noindent
\emph{``Let $\Sigma$ be an fd set defined $\Sigma \!\triangleq\!$ \textsf{h-encode}($\mathcal S$) for some complete structure $\mathcal S$. Then $\Sigma$ is non-redundant but may not be canonical.''}

\begin{proof}
We will show that properties (a-b) of Def. \ref{def:minimal} hold for $\Sigma$ produced by (Alg. \ref{alg:h-encode}) \textsf{h-encode}, but property (c) may not hold. 

At initialization, the algorithm sets $\Sigma\!=\!\varnothing$ and then inserts an fd $\langle X, A\rangle \in \Sigma$ for each $\langle f,\, x \rangle \!\in \varphi_t$ scanned, where $x \mapsto A$ and $X\cap A=\varnothing$. At termination, for all fd's in $\Sigma$ we obviously have $|A|=1$ then property (a) holds. Also, note that $\varphi_t\!: S \!\to Vars(S)$ is, by Def. \ref{def:total-causal-mapping}, a bijection. 

Now, for property (b) not to hold there must be some fd $\langle X, A\rangle \in \Sigma$ that is redundant and then can be found in the closure of $\Gamma \!=\! \Sigma \setminus \langle X, A\rangle$. 
By Lemma \ref{lemma:singleton-rhs}, that can be the case only if $A \subseteq X$ or there is $\langle Y, A\rangle \in\Gamma$ for some $Y$. 
But from $X\cap A=\varnothing$, we have $A \nsubseteq X$; and from $\varphi_t$ being a bijection it follows that there can be no such fd in $\Gamma$. Thus it must be the case that $\Sigma$ is non-redundant, i.e., property (b) holds.

Finally, property (c) does not hold if there can be some fd $\langle X, A\rangle \in \Sigma$ with $Y \!\subset X$ such that $\Gamma=\Sigma \setminus \langle X, A\rangle \cup \langle Y, A\rangle$ has the same closure as $\Sigma$. $\!$That is, 
if we may find $\langle Y, A\rangle \in \Sigma^+$. 
$\!$Now, pick structure $S$ whose ($3 \times 3$) matrix is $A_s \!=\! \{ 1, 0, 0;\, 1, 1, 0;\, 1, 1, 1; \}$ as an instance. Alg. \ref{alg:h-encode} encodes it into $\Sigma\!=\!\{\phi \!\to x_1,\; x_1 \,\upsilon \!\to x_2,\; x_1 \,x_2 \,\upsilon \!\to x_3 \}$. Let $Y\!=\!\{x_1, \upsilon\}$, $Z\!=\!\{x_2\}$ and $X\!=\!YZ$ such that $Y \!\subset\! X$. Note that $x_1 \upsilon \!\to\! x_2 \in \!\Sigma$ can be written as $\langle Y, Z\rangle \in \Sigma$, and $x_1 \,x_2 \,\upsilon \!\to\! x_3 \in \Sigma$ as $\langle YZ, A\rangle \in \Sigma$.  
Apply R5, R0 to derive $\langle Y, A\rangle \in \Sigma^+$, which is sufficient to show that property (c) may not hold.
\end{proof}

\subsection{Proof of Corollary \ref{cor:parsimonious}.}\label{a:parsimonious}
\noindent
``\emph{Let $\Sigma$ be an fd set defined $\Sigma \!\triangleq\!$ \textsf{h-encode}($\mathcal S$) for some complete structure $\mathcal S$. Then $\Sigma$ can be assumed parsimonious with no loss of generality at expense of time that is polynomial in $|\Sigma|\cdot|U|$.}''

\begin{proof}
By Theorem \ref{thm:non-redundant}, we know that any fd set $\Sigma$ produced by \textsf{h-encode} is (Def. $\!$\ref{def:minimal}-a) singleton-rhs and (Def. $\!$\ref{def:minimal}-b) non-redundant. Then let $\Sigma^\prime \triangleq$ \textsf{left-reduce}($\Sigma$) (cf. $\!$Alg. $\!$\ref{alg:left-reduce} in \S\ref{algs}), which is essentially testing, for all fd's $\langle X, Y\rangle \in \Sigma$ and all attributes $A \in X$, whether $\langle X\!\setminus\! A,\; Y\rangle \in \Sigma^+$ Such test is dominated by polynomial-time \textsf{XClosure} (cf. Remark \ref{rmk:linear}). Since it ensures $\Sigma^\prime$ to be (Def. $\!\!$\ref{def:minimal}-c) left-reduced, then $\Sigma^\prime$ must be canonical. 

Moreover, in order to verify it is parsimonious, recall that \textsf{h-encode} requires its input structure $\mathcal S$ to be complete and then (\textsf{COA$_t$}) processes it into a bijection $\varphi_t\!: \mathcal E \to \mathcal V\,$ from equations to variables. That is, each variable $x \in \mathcal V$ is mapped to a relational attribute $x \mapsto B$ and has a unique set of variables $\mathcal V^\prime \subset \mathcal V$ that is encoded to determine it, i.e., $\Sigma \triangleq$ \textsf{h-encode}($\mathcal S$) is such that $\langle Y, B\rangle \in \Sigma$ for exactly one attribute set $Y$ where $\mathcal V^\prime \mapsto Y$. Therefore, $\Sigma^\prime$ must also be parsimonious.
\end{proof}

\subsection{Proof of Proposition \ref{prop:3nf}}\label{a:3nf}
\noindent
\emph{``Let $\boldsymbol R[U]$ be a relational schema,\vspace{-1pt} defined $\boldsymbol R \triangleq$ \textsf{synthe\-size}($\Sigma$) for some  canonical fd set $\Sigma$ on attributes $U$. Then $\boldsymbol R$ preserves $\Sigma$, and is in 3NF but may not be in BCNF.''} 
\begin{proof}
Note that a relation scheme $R_i[S]$ is rendered into $\boldsymbol R$ only if we have an fd $\langle X, Z\rangle \in \Sigma^\prime$ such that $XZ \subseteq S \subseteq U$; and also, for all fd's $\langle Y, W\rangle \in \Sigma^\prime$, there must be some scheme $R_j[T] \in \boldsymbol R$ with with $YW \subseteq T \subseteq U$. As fd set $\Sigma$ is recoverable from $\Sigma^\prime$ by a finite application of (R3) decomposition (cf. line 2 of Alg. \ref{alg:synthesize}), it includes all the fd's projected onto $\boldsymbol R$. Since $\Sigma$ is canonical, by Def. $\!$\ref{def:minimal} it is a cover for every fd set $\Gamma$ such that $\Gamma^+ \!=\! \Sigma^+$. Then 
schema $\boldsymbol R$ synthesized over $\Sigma^\prime$ clearly preserves $\Sigma$. Now we show that schema $\boldsymbol R$ may not be in BCNF. Let Alg. $\!$\ref{alg:synthesize} be applied to canonical $\Sigma$ given in Fig. \ref{fig:h-fdschema} (left). For synthesized relation $R_2[YW]$, where $\,Y\!=\!\{x_1, x_2, x_3, x_5, \upsilon\}$ and $\,W\!=\!\{x_4\}$, note that both fd's $x_1\,x_2\,x_3\,x_5\,\upsilon\!\to x_4$ and $x_1\,x_3\,x_4\,\upsilon \!\to x_5$ hold. But the latter has the form $X \!\to A$, where $A \!\nsubseteq\! X$ but $X \!\not\to \{x_2\}$ then $X \!\not\to YW$. That is, by Def. $\!$\ref{def:nf}, it violates BCNF in $\!R_2$ (then in $\boldsymbol R$). Note that $A\!=\!x_5$ is prime in $\!R_2$, thus 3NF is not violated in $R_2$ by that fd. In fact any $\boldsymbol R$ synthesized by Alg. \ref{alg:synthesize} is in 3NF.

To show that $\boldsymbol R$ is in 3NF, we first disconsider line 2 of (Alg. $\!$\ref{alg:synthesize}) \textsf{synthesize} to study the schema as if synthesized over canonical $\Sigma$ and then extend the proof to $\Sigma^\prime$. $\!$Let $\langle Y, B\rangle \in \Sigma$ be the fd over which scheme $R_k[YB]$ has been synthesized. We have to prove that $R_k$ is in 3NF. By contradiction, suppose there is some $\langle X, A\rangle \in \Sigma^+$, with $XA \!\subseteq YB$, that violates 3NF in $R_k[YB]$. That is, by Def. $\!$\ref{def:nf}, we have $A \nsubseteq X$ but $X \!\not\to YB$ and $A$ is nonprime. Following Ullman \cite[p. $\!$410]{ullman1988}, we have two cases for analysis. If $A=B$, then $X \!\to B$ and, since $A \not\subseteq X$ and $XA \!\subseteq YB$, we have $X \!\subseteq Y$. But as $X \!\not\to Y$, it turns out that $X \!\subset Y$ and $X\!\to B$, while $Y\!\to B$ is in sup\-pose\-dly canonical $\Sigma$. $\!\lightning$. Else ($A\neq B$), we have $A \in Y$. Let $Z \subseteq Y$ be a key for $YB$ then, since $A$ is nonprime, $A \notin Z$. That is, we have $Z \subset Y$ and $Z \!\to B$, while $Y\!\to B$ is in supposedly canonical $\Sigma$. $\!\lightning$. Therefore $R_k[YB]$ (and $\boldsymbol R$ in general) must be in 3NF. 

The extension of this proof to $\Sigma^\prime$ (which is $\Sigma$ after application of R4) is straightforward. For relation $R_k[YW]$, we may have $W\!=\!B_1 B_2 \hdots B_n$, with $Y\!\to W$ in left-reduced, non-redundant $\Sigma^\prime$. Suppose there is some $\langle X, A\rangle \in (\Sigma^\prime)^+$, with $XA \!\subseteq YW$, that violates 3NF for $R_k[YW]$. If $A \!\subset W$, let $Y^\prime=Y\cup (W\!\setminus\! A)$. Then we have $R_k^\prime[Y^\prime\! A]$ under the same properties analyzed above. Else ($A \notin W$), thus we have $A \in Y$ idem. Either way establishes a contradiction, therefore $R_k[YZ]$ (and $\boldsymbol R$ in general) must be in 3NF.
\end{proof}

\subsection{Proof of Lemma \ref{lemma:folding-unique}}\label{a:folding-unique}
\noindent
\emph{``Let $\Sigma$ be a canonical fd set on attributes $U$, and $\langle X, A\rangle \!\in\! \Sigma$ be an fd with $XA \subseteq U$. Then $A^\looparrowright\!$, the attribute folding of $A$ (w.r.t. $\!\Sigma$) exists. Moreover, if $\Sigma$ is parsimonious then $A^\looparrowright\!$ is unique.''}
\begin{proof}
The existance of $A^\looparrowright$ is ensured by the degenerate case where $X = A^\looparrowright\!$, as $X \!\to A$ is itself in $\Sigma^\vartriangleright$ by an empty application of $\{$R0, R5$\}$. If $X \!\to A$ is in fact folded w.r.t. $\!\Sigma$, then the folding of $A$ exists. Else, it is not folded yet $X \!\to A$ is non-trivial because $\Sigma$ is canonical by assumption. Then, by Def. \ref{def:folded} there must be some 
$Y \subseteq U$ with $Y \nsupseteq X$ such that $Y \to X$ and $X \not\to Y$. By Def. \ref{def:ptc}, there is a finite application of rules $\{$R0, R5$\}$ over fd's in $\Sigma$ to derive $Y \!\to X$. 
Then by R2$\,\sim\,$R5 over $X \!\to A$, we have $Y \!\to A$. Although there may be many such (intermediate) attribute sets $Y \subset U$ along the causal ordering satisfying the conditions above, we claim there is at least one that is a folding of $A$. Suppose not. Then, for all such $Y \subset U$, there is some $Y^\prime \subset U$ with $Y^\prime \nsupseteq Y$ such that $Y^\prime \to Y$ and $Y \not\to Y^\prime$, leading to an infinite regress. Nonetheless, in so far as cycles are ruled out by force of Def. \ref{def:folded}, then $\Sigma^+$ must have an infinite number of fd's. But $\Sigma^+$ is finite, viz., bounded by $2^{2|U|}$ (cf. \cite[p. $\!165$]{abiteboul1995}). $\!\lightning$. Therefore the folding of $A$ must exist.

If $\Sigma$ is assumed parsimonious, then by Def. \ref{def:parsimonious} then we have $\langle X, A\rangle \in \Sigma$ for exactly one attribute set $X$. Then, as a straightfoward follow-up of the above reasoning that let us infer the folding existance, note that there must be a single chaining $Y^n \!\xrightarrow{\triangleright} ... \!\xrightarrow{\triangleright} Y^1 \!\xrightarrow{\triangleright} Y^0 \!\xrightarrow{\triangleright} X \!\xrightarrow{\triangleright} A$. Again, as cycles are ruled out by force of Def. \ref{def:folded} and $\Sigma^+$ is finite, then the folding of $A$ is unique.
\end{proof}

\subsection{Proof of Theorem \ref{thm:afolding}}\label{a:afolding}
\noindent
\emph{``Let $\Sigma$ be a parsimonious fd set on attributes $U$, and $A$ be an attribute with $\langle X, A\rangle \in \Sigma$ with $XA \subseteq U$. Then \textsf{AFolding}($\Sigma, A$) correctly computes $A^\looparrowright\!$, the attribute folding of $A$ (w.r.t. $\!\Sigma$) in time $O(n^2)$ in $|\Sigma| \cdot |U|$.''}

\begin{proof}
For the proof roadmap, note that \textsf{AFolding} 
is monotone and terminates precisely when $A^{(i+1)}\!=\!A^{(i)}$, where $A^{(i)}$ denotes the attributes in $A^\star$ at step $i$ of the outer loop. The folding $A^\looparrowright\!$ of $A$ at step $i$ is $A^{(i)} \setminus \Lambda^{(i)}$. We shall prove by induction, given attribute $A$ from fd $X \!\to A$ in parsimonious $\Sigma$, that $A^\star \!\setminus \Lambda$ returned by \textsf{AFolding}($\Sigma, A$) is the unique attribute folding $A^\looparrowright\!$ of $A$. 

$\!\!$(\textbf{Base case}).  
Since $\Sigma$ is assumed canonical with (then) non-trivial $\langle X, A\rangle \!\in \Sigma$ for exactly one attribute set $X$, the algorithm always reaches step $i\!=\!1$, which is our base case. Then $X$ is placed in $A^{(1)}$ and $A$ in $\Lambda^{(1)}$, and we have $A^{(1)}\!=\!XA$ and $\Lambda^{(1)}\!=\!A$. Therefore, $A^{(1)} \setminus \Lambda^{(1)} \!=\! X$, and in fact we have $\langle X, A\rangle \in \Sigma^\vartriangleright$. $\!$For it to be specifically in $\Sigma^\looparrowright\!\! \subset \Sigma^\vartriangleright\!$, it must be folded w.r.t. set $\Gamma$ of consumed fd's at this step, viz., $\Delta^{(1)} \!=\! \{X\!\to A\}$. Since $\langle X, A \rangle \!\in \Sigma^\vartriangleright\!$ and by Def. \ref{def:folded} is in fact folded w.r.t. $\Delta^{(1)}$, we have $A^\looparrowright\!=X$ at step $i\!=\!1$. 

$\!\!$(\textbf{Induction}). Let $i\!=\!k$, for $k\!>\!1$, and assume that $\langle A^{(k)} \!\setminus\Lambda^{(k)}\!,\, A\rangle \in \Sigma^\vartriangleright\!$ is in $\Sigma^\looparrowright\!$, with $A^{(k)} \!\neq\! \Lambda^{(k)}$. $\!$That is, by Lemma \ref{lemma:folding-unique} we know that $A^\looparrowright\!=A^{(k)}\!\setminus\! \Lambda^{(k)}$ is the unique folding of $A$ at step $i\!=\!k$. 
Now, for the inductive step, suppose $Y$ is placed in $A^{(k+1)}$ and $B$ in $\Lambda^{(k+1)}$ because $\langle Y, B \rangle \!\in \Sigma$ and $B \in A^{(k)}$, with $Y \!\cap \Lambda^{(k)}\!=\varnothing$. That is, $A^{(k+1)}=A^{(k)}Y$ and $\Lambda^{(k+1)}\!\!=\!\Lambda^{(k)}B$.
$\!$We must show that $\langle A^{(k+1)} \setminus \Lambda^{(k+1)}\!,\, A\rangle \!\in \Sigma^\looparrowright\!\!$.

$\!$By the inductive hypothesis, we have $(A^{(k)}\setminus \Lambda^{(k)})\!\to A$ in $\Sigma^\looparrowright\! \subset \Sigma^\vartriangleright\!$. Note that such fd implies $\!(A^{(k)}\!B \!\setminus\! \Lambda^{(k)}\!B) \!\to\! A$. Since $B \in A^{(k)}$, we actually have $(A^{(k)}\!\setminus \Lambda^{(k)}\!B) \to A$. Observe that $(A^{(k)}\!\setminus \Lambda^{(k)}\!B) \subseteq (A^{(k)}Y \setminus \Lambda^{(k)}\!B)$. Then, by R0, we must have $(A^{(k)}Y \!\setminus \Lambda^{(k)}\!B) \to (A^{(k)}\!\setminus \Lambda^{(k)}\!B)$. Now, by R5, we get $(A^{(k)}Y \setminus \Lambda^{(k)}\!B) \to A$. Finally, as $A^{(k)}Y=A^{(k+1)}$ and $\Lambda^{(k)}B=\Lambda^{(k+1)}$, we infer $(A^{(k+1)} \!\setminus \Lambda^{(k+1)}) \to A$ is in $\Sigma^\vartriangleright\!$. That has been derived by implication from the inductive hypothesis through a finite application of $\{$R0, R5$\}$. 

Moreover, with the addition of $\,Y \!\to B$ into $\Delta^{(k+1)}\!$, observe that previous $\langle A^{(k)}\setminus \Lambda^{(k)},\, A\rangle\! \in \Sigma^\vartriangleright$ is no longer in $\Sigma^\looparrowright\!$, as it is not folded w.r.t. $\!\Delta^{(k+1)} \supset \langle Y, B\rangle$. In fact, in accordance with Lemma \ref{lemma:folding-unique}, it is replaced by $\langle A^{(k+1)} \setminus \Lambda^{(k+1)}\!,\, A\rangle$ in $\Sigma^\looparrowright\!$, which is folded w.r.t. $\Delta^{(k+1)}$, i.e., $A^\looparrowright\!=A^{(k+1)} \setminus \Lambda^{(k+1)}$.

Finally, as for the time bound, note that 
in worst case, exactly one fd is consumed into $\Delta$ for each step of the outer loop. That is, $|\Sigma|\!=\!n$ is decreased stepwise in arithmetic progression 
such that $n + (n\!-\!1) +  \hdots  + 1 = n\,(n\!-\!1)/2$ scans are required overall, thus Alg. \ref{alg:afolding} is bounded by $O(n^2)$.
\end{proof}

\subsection{Proof of Corollary \ref{cor:folding}}\label{a:folding}
\noindent
\emph{``Let $\Sigma$ be a parsimonious fd set on attributes $U$. Then algorithm \textsf{folding}($\Sigma$) correctly computes $\Sigma^\looparrowright\!$, the folding of $\Sigma$ in time that is $f(n)\,\Theta(n)$ in the size $|\Sigma| \cdot |U|$, where $f(n)$ is the time complexity of (Alg. \ref{alg:afolding}) \textsf{AFolding}.''}
\vspace{3pt}
\begin{proof}
By Theorem \ref{thm:afolding}, we know that sub-procedure (Alg. \ref{alg:afolding}) \textsf{AFolding} is correct and terminates. Then (Alg. \vspace{-2.5pt}\ref{alg:folding}) \textsf{folding} necessarily inserts in $\Gamma$ (initialized empty) exactly one fd $Z \!\xrightarrow{\looparrowright} A$ for each fd $X \!\to A$ in $\Sigma$ scanned. Thus, at termination we have $|\Gamma|\!=\!|\Sigma|$. Again, as \textsf{AFolding} is correct, we know $Z$ is the unique folding of $A$. Therefore it must be the case that Alg. \ref{alg:folding} is correct. Finally, for the time bound, the algorithm iterates in any case over each fd in $\Sigma$, and at each such step \textsf{AFolding} takes time that is $f(n)$. Thus \textsf{folding} takes $f(n)\,\Theta(n)$. But we know from Theorem 
\ref{thm:afolding} and Remark \ref{rmk:linear} that $f(n) \in O(n)$, then it takes $O(n^2)$.
\end{proof}

\subsection{Proof of Proposition \ref{prop:folding-parsimonious}}\label{a:folding-parsimonious}
\noindent
\emph{``Let $\Sigma$ be an fd set, and $\Sigma^\looparrowright\!$ its folding. If $\Sigma$ is parsimonious then so is $\Sigma^\looparrowright\!$.''}
\begin{proof}
Since $\Sigma$ is parsimonious, by Lemma \ref{lemma:folding-unique} we know that, for each fd $\langle X, A\rangle \in \Sigma$, the attribute folding $Z$ of $A$ such that $Z \xrightarrow{\looparrowright} A$ exists and is unique. That is, for no $Y \neq Z$ we have $Y \xrightarrow{\looparrowright} A$. Thus $\Sigma^\looparrowright\!$ := \textsf{folding}($\Sigma$) automatically satisfies Def. \ref{def:parsimonious}, as long as we show it is canonical (cf. Def. \ref{def:minimal}). 

First, note that $\Sigma$ is parsimonious. Then, by Def. \ref{def:parsimonious}, it must be canonical and, by Def. \ref{def:minimal}, it is singleton-rhs, non-redundant and left-reduced.
Now, consider by Lemma \ref{lemma:folding-unique} that \textsf{AFolding} builds a bijection mapping each\vspace{-2pt} $\langle X, A\rangle \in \Sigma$ to exactly one $\langle Z, A\rangle \in \Sigma^\looparrowright\!$ such that $Z \xrightarrow{\looparrowright} A$. As $\Sigma$ is singleton-rhs, it is obvious that $\Sigma^\looparrowright\!$ is as well. Also, the bijection implies $|\Sigma^\looparrowright\!| = |\Sigma|$. Since $\Sigma$ is non-redundant, then by Lemma \ref{lemma:singleton-rhs} so is $\Sigma^\looparrowright\!$.

Besides, suppose by contradiction that $\Sigma^\looparrowright\!$ is not left-reduced. Then there is some fd $Z \!\to A$ in $\Sigma^\looparrowright\!$ with $W \subset Z$ such that $W \!\to A$ is in $(\Sigma^\looparrowright\!)^+$. But as $\Sigma$ is parsimonious, $X \!\to A$ is the only fd with $A$ in its rhs and we know $W \nsubseteq X$ by $\Sigma$ being left-reduced. Then for $W \!\to A$ to be in $(\Sigma^\looparrowright\!)^+$, it must be the case that $W \xrightarrow{\triangleright} X$ and then $W \xrightarrow{\triangleright} A$. However, as $Z \xrightarrow{\triangleright} A$ is folded and $W \subset Z$, then $W \xrightarrow{\triangleright} A$ must be folded as well, i.e., the folding of $A$ is not unique. $\lightning$.
Therefore $\Sigma^\looparrowright\!$ must be parsimonious.
\end{proof}

\subsection{Proof of Theorem \ref{thm:bcnf}}\label{a:bcnf}
\noindent
\emph{``Let $\boldsymbol R[U]$ be a relational schema, defined $\boldsymbol R \triangleq$ \textsf{synthe\-size}($\Sigma^\looparrowright\!$) where $\Sigma^\looparrowright\!$ be the folding of parsimonious fd set $\Sigma$ on attributes $U$. We claim that $\boldsymbol R$ is in BCNF, is minimal-cardinality and preserves $\Sigma^\looparrowright\!$.''}
\begin{proof}

Verification of dependency preservation w.r.t. $\Sigma^\looparrowright\!$ follows just the same logic as in Proposition \ref{prop:3nf} except that the input fd set is now $\Sigma^\looparrowright\!$, the folding of parsimonious $\Sigma$ (known to have more specific properties). It must be the case that $\boldsymbol R$ preserves $\Sigma^\looparrowright\!$.

Now we concentrate on the BCNF property. 
We shall use Lemma \ref{lemma:check-bcnf}, originally from Osborn \cite{osborn1979},\footnote{Cf. also Ullman \cite[p. $\!$403]{ullman1988}.} to check it in a convenient way.  
$\!$Let $\langle Y, W\rangle \in (\Sigma^\looparrowright\!)^+$ be the fd over which scheme $R_k[YW]$ was synthesized. 
We shall prove that $R_k$[YW] must be in BCNF. By Proposition \ref{prop:folding-parsimonious}, $\Sigma^\looparrowright\!$ is parsimonious (then canonical), and then by Lemma \ref{lemma:check-bcnf} we only need to check for fd violations in $\Sigma^\looparrowright\!$, not $(\Sigma^\looparrowright)^+\!$. 
By contradiction, suppose there is some $\langle X, A\rangle \in \Sigma^\looparrowright\!$, with $XA \!\subseteq YW$, that violates BCNF in $R_k[YW]$. That is, by Def. $\!$\ref{def:nf}, we have $A \nsubseteq X$ but $X \!\not\to YW$. Then, as $Y$ is a key for $R_k[YW]$ and $\langle X, A\rangle \in \Sigma^\looparrowright\!$, we have $X \!\nsubseteq Y$ and then $X$ must (at least) partly intersect with $W$. 
So, let $X\!=\!ST$ for some $T \!\neq\! \varnothing$ with $T \subset W$, and $Y\!=\!SZ$ for some $Z \!\neq\! \varnothing$ with $Z \cap T \!=\! \varnothing$. By assumption, $\langle ST, A\rangle \in \Sigma^\looparrowright\!$. That is, by Def. $\!$\ref{def:folded}, there can be no $V \nsupseteq ST$ such that $V \to ST$ and $ST \not\to V$. Now, take $V\!=Y\!=SZ$. Note that $Z \cap T \!=\! \varnothing$ then $SZ \nsupseteq ST$. Since $T \subset W$ and $SZ \!\to W$, by (R3) decomposition we have $SZ \!\to T$ and then, by (R0) reflexivity, we get $SZ \!\to ST$. But ($X \not\to Y$) $ST \not\to SZ$. That is, $\langle ST, A\rangle \notin \Sigma^\looparrowright\!$. $\!\lightning$. Thus $R_k[YW]$ (and $\boldsymbol R$ in general) must be in BCNF.

For the minimality, note that any two schemes $R_{i}[YW]$, $R_{j}[XZ]$ are rendered by \textsf{synthesize} into $\boldsymbol R$ iff we have fd's $\langle Y, W\rangle, \langle X, Z\rangle \in (\Sigma^\looparrowright)^\prime$ and $X \not\leftrightarrow Y$, i.e., it is not the case that both $X \!\to Y$ and $Y \!\to X$ hold in $\!(\Sigma^\looparrowright)^+$. Now, to prove that $\boldsymbol R$ is minimal-cardinality, we have to find that merging any such pair of arbitrary schemes shall hinder BCNF in $\boldsymbol R$. In fact, take $\boldsymbol R^\prime := \boldsymbol R \setminus (R_i[YW] \cup R_j[XZ]) \cup R_k[YWXZ]$. As $X \not\leftrightarrow Y$, then neither $X$ nor $Y$ can be a superkey for $R_k$, i.e., $R_k$ is not in BCNF.
\end{proof}

\begin{mylemma}
$\!$Let $\boldsymbol R$[U] be a relational schema and $\Sigma$ a  
canonical fd set on attributes $U$. Then, $\boldsymbol R$ can be verified to be in BCNF by checking for violations w.r.t. $\!\Sigma$ only (not w.r.t. $\Sigma^+$).
\label{lemma:check-bcnf}
\vspace{-2pt}
\end{mylemma}	
\begin{proof}
See Osborn \cite{osborn1979}.
\end{proof}

\subsection{Proof of Proposition \ref{prop:lossless}}\label{a:lossless}
\noindent
\emph{``Let $\boldsymbol R[U] \triangleq$ \textsf{synthesize}($\Sigma^\looparrowright\!$) be a relational schema with $|\boldsymbol R|\geq 2$, where $\Sigma^\looparrowright\!$ is the folding of a parsimonious fd set $\Sigma \triangleq$ \textsf{h-encode}($\mathcal S$) on attributes $U$. Then $\boldsymbol R$ has a lossless join (w.r.t. $\Sigma^\looparrowright\!$) iff, for all $R_i[XZ] \in \boldsymbol R[U]$ with key constraint $X \to Z$ for $XZ \subseteq U$, we have $X \to U$ or there is some scheme $R_j[YW] \in \boldsymbol R[U]$ with key constraint $Y \to W$ for $YW \subseteq U$ such that $X \subset Y$.''}
\begin{proof}
We use Lemma \ref{lemma:lossless} from Ullman \cite[p. $\!$397]{ullman1988}, which gives a convenient necessary and sufficient condition for the lossless join property.
By Lemma \ref{lemma:lossless}, any pair\vspace{2pt} $R_i[XZ]$, $R_j[YW] \in \boldsymbol R[U]$ with $XZYW \subseteq U$ have a lossless join w.r.t. $\!\pi_{XZYW}(\Sigma^\looparrowright\!)$, iff $(X \cap Y) \to (XZ \setminus YW)$ or $(X \cap Y) \to (YW \setminus XZ)$ hold in $\!\pi_{XZYW}((\Sigma^\looparrowright\!)^+)$.

Now, let $R_i[XZ], R_j[YW] \!\in\! \boldsymbol R$ be arbitrary schemes. By Theorem \ref{thm:bcnf}, $\boldsymbol R$ is in BCNF. Then, if $R_i[XZ]$ and $R_j[YW]$ are two different schemes in $\boldsymbol R$, we must have $XZ \setminus W = XZ$ and $YW \setminus Z = YW$. That is, we can write the condition imposed by Lemma \ref{lemma:lossless} as $(X \cap\, Y) \to (XZ \setminus Y)$ or $(X \cap Y) \to (YW \setminus X)$.

 Suppose $X \subset Y$. That is, we have $X \cap Y=X$ and $X \to XZ$ then obviously $X \to XZ\setminus Y$. That is, $R_i[XZ]$ and $R_j[YW]$ have a lossless join. Now, suppose rather that $X \to U$. Then $X \to Y$ but $Y \not\to X$ and $Y \neq X$, otherwise $R_i[XZ]$ and $R_j[YW]$ would have been merged by \textsf{synthesize}. Moreover, by (R2) transitivity we have $X \to W$. But for some $B_i \in W$ we must have had $Y \xrightarrow{\looparrowright} B_i$ in $\Sigma^\looparrowright\!$. Since we do have $X \to Y$ and $Y \not\to X$, then by Def. $\!$\ref{def:folded} it must be the case that $X \supset Y$. Therefore $X \cap Y=Y$ and $Y \to YW$, then obviously $Y \to YW\setminus X$. Since $R_i$ and $R_j$ are picked arbitrarily and the natural join operator is associative, $\bowtie_{R_\ell \in \boldsymbol R} \pi_{R_\ell}(r)$ must be lossless.

For the converse, suppose rather that $X \not\to U$ and for all such $R_j[YW]$ we have $X \not\subset Y$. Since $X\neq Y$, we actually have $X \not\subseteq Y$. Now, let $X=ST$ and $Y=SV$ for some $T\neq V$ and $T\neq \varnothing$. That is, $X \cap Y=S$, and (i) $R_i \setminus R_j = STZ \setminus SVW=TZ$ but we have $S \not\to Z$ (otherwise $T$ would be extraneous in key $X$); and (ii) $R_j \setminus R_i = VW \setminus TZ=VW$ but we have $S \not\to V$ likewise. Then there can be no $R_j[YW]$ which $R_i[XZ]$ has a lossless join with, and thus $\boldsymbol R[U]$ cannot have a lossless join.
\end{proof}

\begin{mylemma}
\label{lemma:lossless}
Let $\Sigma$ be a set of fd's on attributes $U$, and $R_i[S], R_j[T] \in \boldsymbol R[U]$ be relation schemes with $ST \subseteq U$; and let $\pi_{ST}(\Sigma)$ be the projection of $\Sigma$ onto $ST$. Then $R_i[S]$ and $R_j[T]$ have a lossless join w.r.t. $\pi_{ST}(\Sigma)$ iff $(S \cap\, T) \to (S \setminus T)$ or $(S \cap\, T) \to (T \setminus S)$ hold in $\pi_{ST}(\Sigma^+)$.
\end{mylemma}
\begin{proof}
See Ullman \cite[p. $\!$397]{ullman1988}.
\end{proof}

\subsection{Proof of Proposition \ref{prop:u-factor}}\label{a:u-factor}
\noindent
\emph{``Let $H_k^\ell[XZ] \in \boldsymbol H_k[U]$ be an exogenous relation with (violated) key constraint $\langle X, Z\rangle \in \Sigma_k^\looparrowright\!$ for $XZ \subseteq U$. $\!$Then,
\begin{itemize}
\item[(a)] for any pair $Y_k^i[\,V_iD_i\,|\,X A_i\,],\; Y_k^j[\,V_jD_j\,|\,X A_j\,]$ of u-factor projections of $H_k^\ell$, they are independent.
\item[(b)] the join $\bowtie_{i=1}^m \! Y_k^i[\,V_iD_i\,|\,X A_i\,]$ of all u-factor projections of $H_k^\ell$ is lossless w.r.t. $\pi_{X A_1\,A_2\,...\,A_m}(\Sigma_k^\looparrowright\!)$.''
\end{itemize}}
\vspace{3pt}
\begin{proof}
$\!$We prove the claims in separate as follows.

(a) Suppose not. Then we have either $A_i \!\to A_j$ or $A_j \!\to A_i$. Let $G_i$ be such that $A_i \in G_i \subseteq Z$. Then, since $Y_k^i[\,V_iD_i\,|\,X A_i\,]$ is a u-factor projection of $H_k^\ell[XZ]$, by Def. $\!$\ref{def:u-factor} there can be no $C \in Z \setminus G_i$ with $A_i \!\to C$ or $C \!\to A_i$. Now, take $G_j$ such that $A_j \in G_j \subseteq Z$. Since $G_i, G_j \subseteq Z$ are learned to be maximal groups that have to be disjoint in $Z$ (cf. Problem \ref{prob:u-learning}), we must have such $A_j \in Z \setminus G_i$ with $A_i \!\to A_j$ or $A_j \!\to A_i$. $\lightning$.

(b) First, note that the lossless join property is considered w.r.t. $\!\pi_{X A_1\,A_2\,...\,A_p}(\Sigma_k^\looparrowright\!)$.
By Lemma \ref{lemma:lossless},\vspace{-2pt} we know that any pair $Y_k^i[\,V_iD_i\,| X A_i],\, Y_k^j[\,V_jD_j\,| X A_j]$ of u-factor projections will have a lossless join w.r.t. $\pi_{X A_i\,A_j}(\Sigma_k^\looparrowright\!)$ iff $(XA_i \,\cap\, XA_j) \to (XA_i \setminus XA_j)$ or $(XA_i \,\cap\, XA_j) \to (XA_j \setminus XA_i)$ hold in $\pi_{X A_i\,A_j}((\Sigma_k^\looparrowright\!)^+)$. By Def. \ref{def:u-factor}, we must have $XA_i \cap XA_j=X$, and $XA_i \setminus XA_j = A_i$. In fact $X \!\to A_i$ is in $\pi_{X A_1\,A_2\,...\,A_p}(\Sigma_k^\looparrowright\!$), thus $Y_k^i$ and $Y_k^j$ have a lossless join w.r.t. such fd set projection. Then, from the associativeness of the join, all the u-factor projections taken together must have a lossless join.
\end{proof}

\subsection{Proof of Theorem \ref{thm:u-propagation}}\label{a:u-propagation}
\noindent
\emph{``Let $H_k^q[Z_q V] \in \boldsymbol H_k$ be an endogenous relation with (violated) key constraint $\langle Z_q, V\rangle\! \in \Sigma_k^\looparrowright\!$, and $Y_k^r[\overline{V_iD_i}\,|\,Z_qT]$ be a predictive projection of $H_k^q$ w.r.t. $\Gamma_k^\looparrowright\!$ defined by formula (\ref{eq:u-propagation}) with $V \!\supseteq T$. $\!$\mbox{We claim that $Y_k^r$ correctly} captures all the uncertainty present in $H_k^q$ w.r.t. $\Gamma_k^\looparrowright\!$.''}
\begin{proof}
For the proof roadmap, consider that violated key constraint $\langle Z_q, V\rangle \in \Sigma_k^\looparrowright\!$ is `trivially' repaired `4C' by using special attribute `trial id' \textsf{tid} provisionally, viz., take $Z_q^\prime \!=\! Z_q \cup \{\textsf{tid}\}$ and then $Z_q^\prime \!\to V$ holds in endogenous relation $H_k^q[Z_q^\prime V]$. Exogenous relations of form $H_k^\ell[X_\ell\, L]$ are trivially repaired likewise, viz., $X_\ell^\prime\! := X_\ell \cup \{\textsf{tid}\}$ to get clean $H_k^\ell[X_\ell^\prime\, L]$, so that we have (i) a clean loading of trial data, and (ii) a (correct) natural join of tuples from exogenous and endogenous relations. In such setting `4C,' \textsf{tid} works for every endogenous tuple as a surrogate to the exogenous tuple(s) that are in its causal chain. In fact, what\vspace{-1pt} u-propagation does in synthesis `4U' is to replace \textsf{tid} by some set $\overline{V_iD_i}$ of pairs of condition columns such that, for $Z_q^{\prime\prime} \!= \overline{V_iD_i}Z_q$, we know that $Z_q^{\prime\prime} \!\to T$ holds in U-relation $Y_k^r[Z_q^{\prime\prime}\, T]$, a predictive projection of $H_k^q[Z_q^\prime V]$ w.r.t. $\Gamma_k^\looparrowright\!$. By Def. \ref{def:u-propagation}, query formula (\ref{eq:u-propagation}) we outline below accordingly defines exactly such replacement, and what we shall prove by construction is that it does it correctly w.r.t. $\Gamma_k^\looparrowright\!$.
\begin{eqnarray*}
Y_k^r \,:=\, \pi_{Z_qT}(\, \sigma_{\upsilon=k}(Y_0) \bowtie \pi_{X^\prime}(J) \bowtie \pi_{Z_q^\prime T}(H_k^q)\,)
\end{eqnarray*}
\noindent
where $X^\prime \!=\! \textsf{tid} \,\cup X$ and $Z_q^\prime \!=\! \textsf{tid} \,\cup Z_q$.

We start by considering (as the join is associative) sub-query $\sigma_{\upsilon=k}(Y_0) \bowtie \pi_{Z_q^\prime T}(H_k^q)$. By the rewriting rule for the product on U-relations (cf. \S\ref{subsec:u-relations}), it is rewritten to be applied over $\sigma_{\upsilon=k}(Y_0[V_0D_0\,|\,\phi\,\upsilon])$ and $\pi_{Z_q^\prime T}(H_k^q)$, where $\upsilon \in Z_q^\prime$. Moreover, since $\Sigma_k$ is asumed to have empirical grounding, we must have $\phi \in Z_q^\prime$ as well. The partial result set then is $Q_1[V_0D_0\,|\,Z_q^\prime T]$.

Now we consider $\pi_{X^\prime}(J)$, where $J$ is a join sub-query defined according to Def. \ref{def:u-propagation}.
Before we proceed to examine $J$, though, recall that $Y_k^r[\overline{V_iD_i}\,| Z_qT]$ is by assumption a predictive projection of $H_k^q[Z_q V]$. Then there must be $S \supseteq Z_q \setminus \{\phi\}$, where $\langle S, T\rangle \in (\Gamma_k^\looparrowright\!)^+$ is the key constraint of sketched scheme $R[ST]$ with $V \supseteq T$. Note that $S$ contains the u-factors for $T$, unfolded out of their compact representation by $\phi \in Z_q$ such that $\langle Z_q, V\rangle \in (\Sigma_k^\looparrowright\!)^+$. 
Then, by Def. \ref{def:u-propagation}, and taking advantage again of the associativeness of the join, we have: 
\begin{eqnarray*}
J \,=\;\; (\,\bowtie_{H_k^\ell \in \mathsf M} \! H_k^\ell \,) \bowtie (\, \bowtie_{Y_k^j \in \mathsf M(H_k^\ell)} \! (Y_k^j) \,).
\end{eqnarray*}
\noindent
where $\textsf{M}$ is the mapping, from each exogenous relation of form $H_k^\ell[X_\ell^\prime\, L]$ with $(L \,\cap\, S) \neq \varnothing$, to the set of all their u-factor projections $Y_k^j[V_jD_j | X_\ell\, A_j]$ such that we have $A_j \!\in (L \cap S)$ and $X \!= \bigcup_{H_k^\ell \in \mathsf M}X_\ell$. 

Note that, by Proposition \ref{prop:u-factor}, the join $\bowtie_{Y_k^j \in \mathsf M(H_k^\ell)} \! (Y_k^j)$ is lossless. 
Thus, by the rewriting on U-relations  
another partial result set is $\pi_{X^\prime}(J) = Q_2[\overline{V_jD_j} | X^\prime]$, where $\overline{V_jD_j}$ is the set of\vspace{1pt} pairs of condition columns containing the random variables associated with all the empirical u-factors in $S$ and that we want to propagate into $Y_k^r[\overline{V_iD_i}\,| Z_qT]$, i.e., $\overline{V_iD_i} = V_0D_0 \cup \overline{V_jD_j}$.
In fact, back to the formula (\ref{eq:u-propagation}),
\begin{eqnarray*}
Y_k^r \,:=\, \pi_{Z_q T} ( Q_1[V_0D_0\,|\,Z_q^\prime T]  \bowtie Q_2[\overline{V_jD_j} | X^\prime ].
\end{eqnarray*}

We expect this join to be lossless. By Proposition \ref{prop:lossless} that will be the case iff $X^\prime \subset Z_q^\prime$, i.e., $X \subset Z_q$, which we show next. In fact, recall that we have $S \to T$ in $(\Gamma^\looparrowright)^+$ and $Z_q \!\to T$ is an $\upsilon$-fd in $(\Sigma^\looparrowright\!)^+$. Let $B \in T$, then $S \!\to B$ is in $\Gamma^\looparrowright$ and $Z_q \!\to B$ is in $\Sigma^\looparrowright\!$. That is, $Z_q \xrightarrow{\looparrowright} B$ is folded. For each key constraint $X_\ell \!\to A_j$ such that $X_\ell$ was built into $X$, note that  $A_j \in S$ and 
$X_\ell \!\to A_j$ is a $\phi$-fd in $\Sigma^\looparrowright\!$, i.e., $X_\ell \xrightarrow{\looparrowright} A_j$ is folded as well. That is, $X_\ell$ must have replaced $A_j$ into $Z_q$. That is, we must have $X \subset Z_q$.

Therefore the join $Q_1[V_0D_0\,|\,Z_q^\prime T]  \bowtie Q_2[\overline{V_jD_j} | X^\prime ]$ is lossless w.r.t. $\Gamma_k^\looparrowright\!$ and we have the result set materialized into $Y_k^r[V_0D_0 \overline{V_jD_j} \,\,Z_q T]$, which has exactly the u-factors indicated by $\langle S, T\rangle \in (\Gamma_k^\looparrowright\!)^+$. That is, u-propagation into $Y_k^r$ according to formula (\ref{eq:u-propagation}) must be correct.
\end{proof}

\section{Detailed Comments}

\subsection{Comments on Conjecture \ref{conj:lossless}}\label{a:conjecture}

As previsouly suggested in \S\ref{sec:encoding}, our framework brings forth a translation of SEM's concepts into the language of fd's. In general, that seems to be an interesting result, as an expression of properties of deterministic hypothesis in relational theory. 
Along these lines, say, Proposition \ref{prop:exo-endo} below relates to Def. \ref{def:exo-endo}.

\begin{myprop}
Let $\Sigma$ be an fd set over attributes $U$, defined $\Sigma \!:=$\textsf{h-encode}($\mathcal S$) for some complete structure $\mathcal S(\mathcal E, \mathcal V)$, and $x \mapsto A$ for some $x \in \mathcal V$ and $A \in U$. Then $A$ is exogenous (endogenous) iff $x$ is exogenous (endogenous).
\label{prop:exo-endo}
\end{myprop}
\begin{proof}
This shall be straightforward from Def. \ref{def:exogenous} and Def. \ref{def:exo-endo}, by observing (Alg. \ref{alg:h-encode}) \textsf{h-encode}.
\end{proof}
\noindent
Conjecture \ref{conj:lossless} works in the converse direction towards more complete results in terms of the equivalence between SEM's and our design-theoretic framework.

\vspace{3pt}
\noindent
\emph{``The lossless join property is reducible to the structure $\mathcal S$ given as input to the pipeline.''} 
\vspace{3pt}

We leave to future work to show that the condition for the lossless property (cf. Proposition \ref{prop:lossless}) in fact translates into the structure $\mathcal S$ given as input to the pipeline. Note that we have technical means to decide/test for the lossless join efficiently: Ullman gives a polynomial-time algorithm for that \cite[p. $\!$406]{ullman1988}, and Proposition \ref{prop:lossless} gives an even more specific condition (as the schema is known to be in BCNF) that can be fruitfully exploited in that sense. Thus, we shall be able to decide whether a given structure $\mathcal S$ is ``fully interconnected'' by processing its folding and testing its rendered schema for the lossless join.

Moreover, consider fd set $\Gamma$ from Example \ref{ex:lossless}. Arguably, it may not make sense in practice to keep those two pairs of variables in the same hypothesis/model --- the modeler user would likely strive for each hypothesis to form a single, cohesive piece of scientific inquiry. That is, 
we expect in practice for the majority of hypotheses given/extracted that they shall have a lossless join. In fact, we find some evidence for that from our initial experiments in a real-world use case (cf. \S\ref{subsec:applicability}).

\end{document}